\documentclass[11pt]{article}

\usepackage[margin=0.95in]{geometry}
\usepackage{amsmath,amssymb,amsthm,mathtools}
\usepackage{graphicx}
\usepackage{booktabs}
\usepackage[authoryear,round]{natbib}
\usepackage[colorlinks=true,linkcolor=blue,citecolor=blue,urlcolor=blue]{hyperref}
\usepackage{cleveref}
\usepackage{setspace}
\theoremstyle{plain}
\newtheorem{theorem}{Theorem}[section]
\newtheorem{lemma}[theorem]{Lemma}
\newtheorem{proposition}[theorem]{Proposition}
\newtheorem{corollary}[theorem]{Corollary}

\theoremstyle{definition}
\newtheorem{remark}{Remark}[section]

\newcommand{\Sa}{S_a}
\newcommand{\Su}{S_u}
\newcommand{\Ia}{I_a}
\newcommand{\Iu}{I_u}
\newcommand{\Rzero}{\mathcal{R}_0}
\newcommand{\Sfun}{\mathcal{S}}
\newcommand{\eqR}{\mathbb{R}}
\newcommand{\Gam}{\Gamma}
\newcommand{\kzo}{\kappa_1^{0}}
\newcommand{\kzt}{\kappa_2^{0}}
\newcommand{\dzo}{\delta_1^{0}}
\newcommand{\dzt}{\delta_2^{0}}

\usepackage{titlesec}
\titlespacing*{\section}
  {0pt}{1.5ex plus 0.3ex minus 0.2ex}{0.6ex plus 0.2ex}

\titlespacing*{\subsection}
  {0pt}{1.2ex plus 0.3ex minus 0.2ex}{0.5ex plus 0.2ex}

\titlespacing*{\subsubsection}
  {0pt}{1.0ex plus 0.2ex minus 0.2ex}{0.4ex plus 0.1ex}

\title{Long-term Coexistence of Epidemics and Risk Awareness:\\
Impacts of Adaptive Human Response and Fatigue}

\author{Mozzamil Mohammed$^1$, Abdallah Alsammani$^{2,\ast}$\\[1ex]
\small $^1$College of Business, Jumeira University, Dubai, United Arab Emirates\\
\small $^2$Department of Mathematical Sciences,
\small Delaware State University, Dover, DE, USA\\[1ex]
\small $^\ast$Correspondence: \texttt{aalsammani@desu.edu}}

\date{}

\begin{document}
\maketitle

\begin{abstract}
\noindent
Human behavior shapes epidemic dynamics, yet most models represent it by rescaling transmission, conflating behavior with biology and removing the memory carried by sustained protective behavior. We develop and analyze a susceptible-infected-recovered-aware model in which prevalence-driven awareness, eroded by fatigue, moves susceptible and infectious individuals between accessible and inaccessible states rather than modifying transmission probabilities. The basic reproduction number is $\Rzero=\beta/(\gamma+\mu)$, independent of behavioral parameters, and a unique endemic equilibrium exists exactly when $\Rzero>1$, excluding backward bifurcation. Infection and awareness are uniformly persistent whenever $\Rzero>1$. We derive an exact decomposition of the reduction in endemic burden into susceptible withdrawal and infectious sequestration, establish the effect of each behavioral parameter, and show that the endemic state depends on awareness response and fatigue only through their ratio, with prevalence vanishing linearly as fatigue vanishes. In the instantaneous-relaxation limit, the model reduces to a transmission-modulating formulation whose endemic equilibrium is locally asymptotically stable. At finite relaxation rates, however, numerical stability and bifurcation analyses reveal a Hopf transition and self-sustained epidemic cycles without imposed delay. Relative to the SIR benchmark, awareness lengthens transient inter-epidemic periods and slows damping rather than creating the underlying transient oscillations.
\end{abstract}
\vspace{2ex}

\section{Introduction}
\label{sec:intro}

Epidemic trajectories are shaped jointly by pathogen biology and by how
populations perceive and respond to risk
\citep{Funk2009,Fenichel2011,Poletti2012,OGara2025}.
During outbreaks, people reduce contacts when perceived risk is high and
resume normal activity as concern fades \citep{Poletti2009,Engeli2025}.
Because these adjustments alter who is available for contact, they influence
how fast infection spreads, how large outbreaks become, and how long
transmission persists. Representing them mechanistically is therefore central
to the theory of epidemics and to the design of public health responses
\citep{Funk2010,Verelst2016,Wang2015}.

The dominant modeling convention is to let behavior act on the transmission coefficient. An information index or awareness variable is constructed from current or recent prevalence, and the force of infection is multiplied by a decreasing function of that variable \citep{Funk2009,dOnofrio2009,Fenichel2011,Weitz2020,Juher2023,Mahmud2025}. This convention is analytically convenient and reproduces the first-order observation that outbreaks are damped by risk perception. It carries three structural consequences that are rarely made explicit. First, behavior and biology are merged into a single coefficient, so the model cannot separate a change in who participates in transmission from a change in how efficiently a contact transmits. Second, the behavioral response is memoryless at the population level: when awareness falls, every individual is instantaneously returned to full participation, and no record of past withdrawal is retained. Third, because protective behavior is not a state that individuals occupy, the model has no notion of the duration of protective behavior and hence no behavioral residence time.

An alternative is to treat protective behavior as an occupied state. Compartmental representations of fear, self-isolation, or self-initiated distancing were introduced by \citet{Epstein2008} and by \citet{Perra2011}, whose self-initiated behavior-change class is entered at a prevalence-dependent rate and left at a relaxation rate; the pair $\kappa_1(A),\kappa_2(A)$ used below is a close relative of that construction. Related mechanisms have been studied mainly through simulation and on networks \citep{Granell2013,Wang2015,Wang2016}, and through adaptive rewiring, in which avoidance is realized as a change of contact topology rather than of compartment \citep{Gross2006}. Awareness-driven behavior is likewise known to reshape epidemic curves into plateaus, shoulders, and oscillations, and to bias forecasts that ignore it \citep{Eksin2019,Weitz2020,Alsammani2025Vaccination}. These formulations retain the feature that matters here: an individual who withdraws leaves the transmitting population for a period governed by a rate, not by the instantaneous value of a signal. What has been missing is an analytical treatment that identifies precisely which epidemiological conclusions depend on this structural choice, and a demonstration that the two representations are not merely different parameterizations of the same dynamics.

Protective behavior is not sustained indefinitely. Behavioral fatigue, also called pandemic fatigue, denotes the progressive erosion of precautionary behavior under persistent risk, driven by habituation, economic pressure, and political conflict \citep{Harvey2020,Jorgensen2022,Chen2024}. Fatigue returns withdrawn individuals to the transmitting population, thereby closing the feedback loop between disease and behavior. In transmission-modulating models, fatigue enters as a decay rate of the information index; in a state-based model, it also determines how long a cohort remains withdrawn and therefore assumes a dynamical role that has no counterpart in the former class.

Three gaps motivate this work. The \emph{conceptual} gap is that awareness-driven accessibility, that is, awareness-dependent movement between behavioral states that do and do not transmit, has not been separated cleanly from awareness-driven suppression of transmissibility, so it is not known which epidemic phenomena require the former. The \emph{mathematical} gap is that behavioral epidemic models are typically analyzed only up to the reproduction number and local stability; global uniqueness of the endemic state, persistence of the behavioral variable itself, and dependence of the endemic burden on behavioral parameters are rarely established. The \emph{biological} gap is that the mechanisms sustaining long-term coexistence of low-level infection with continued behavioral vigilance and those generating recurrent waves in the absence of seasonal forcing remain incompletely resolved.

We state our increment over this literature precisely, since neither behavior-change compartments nor awareness-driven epidemic suppression is new. What is new here is the combination of an autonomous awareness equation with an explicit fatigue term, a symmetric accessibility structure in \emph{both} the susceptible and infectious classes, and vital dynamics that permit a genuine endemic analysis, assembled in a form that remains analytically tractable. That tractability is what allows the endemic state to be characterized exactly rather than explored numerically, and what allows the accessibility formulation to be separated from the transmission-modulating one by a theorem rather than by comparison of simulations.

We formulate and analyze a susceptible-infected-recovered-aware (SIRA) model in which a population-level awareness variable, generated by prevalence and eroded by fatigue, governs transitions of susceptible and infectious individuals between accessible and inaccessible states. Transmission occurs only between accessible individuals, and no transmission probability depends on awareness. Our results are as follows.
\begin{enumerate}
\itemsep0.15em
\item \emph{Threshold neutrality.} $\Rzero=\beta/(\gamma+\mu)$ is independent of all behavioral parameters (\Cref{sec:R0}), and the disease-free equilibrium is globally asymptotically stable when $\Rzero<1$ (\Cref{thm:global_dfe}).
\item \emph{Existence and uniqueness.} Equilibrium analysis reduces to a scalar equation with a strictly decreasing left-hand side, giving a unique endemic equilibrium when $\Rzero>1$ and none otherwise (\Cref{thm:endemic}). Backward bifurcation and bistability are excluded for all admissible parameters.
\item \emph{Robust coexistence.} Infection and awareness are uniformly persistent when $\Rzero>1$ (\Cref{thm:persistence}); coexistence therefore does not depend on which attractor the system settles on.
\item \emph{Exact burden decomposition.} The reduction of endemic prevalence relative to the SIR benchmark splits additively into susceptible withdrawal and infectious sequestration (\Cref{thm:decomp}), and the direction of the effect of every behavioral parameter is determined analytically (\Cref{thm:monotone}).
\item \emph{A single behavioral control parameter.} The endemic state depends on awareness response $\rho$ and fatigue $\omega$ only through the ratio $\theta=\omega/\rho$, and $I^{*}=\theta A_{\infty}/(1-A_{\infty})+o(\theta)$ as $\theta\to0$ (\Cref{thm:theta,thm:lowfatigue}). Behavioral fatigue is thus the only obstruction to driving prevalence to zero, even though awareness never changes the invasion threshold.
\item \emph{Accessibility is not transmission modulation.} In the limit of instantaneous behavioral relaxation, the model reduces to the classical transmission-modulating formulation described in Proposition~\ref{prop:fastlimit}, whose endemic equilibrium we prove to be locally asymptotically stable for every admissible parameter set (\Cref{thm:memoryless}). At finite relaxation rates, numerical eigenvalue and bifurcation analyses show that the endemic equilibrium of the SIRA model can lose stability through a Hopf transition, with direct simulations showing stable self-sustained epidemic cycles (\Cref{sec:hopf}). Finite behavioral residence times can therefore generate recurrent epidemic dynamics without an exogenous delay, memory kernel, or seasonal forcing.
\item \emph{Separation from demographic waves.} The damped multi-wave transients are compared against the no-awareness control obtained at $\rho=0$, which is the classical SIR model with the same demography and which already oscillates. Awareness lengthens the period of those waves and slows their damping; only the self-sustained cycles are without a counterpart in the control (\Cref{sec:sircontrol}).
\end{enumerate}
Item~6 is the sharpest statement of the difference between the two representations: they share $\Rzero$, share the qualitative direction of every endemic comparative static, and differ in the stability of the endemic state.

\Cref{sec:model} presents the model, \Cref{sec:analysis} the analysis, \Cref{sec:results} the numerical results, and \Cref{sec:discussion} the discussion.

\section{Model}
\label{sec:model}

\subsection*{States}

The population is normalized to unity and partitioned into susceptible, infectious, and recovered classes, $S+I+R=1$. Susceptible and infectious individuals are further classified by \emph{accessibility}, the behavioral state that determines whether an individual participates in social contact: $S=\Sa+\Su$ and $I=\Ia+\Iu$, where the subscripts $a$ and $u$ denote accessible and inaccessible individuals. Accessible individuals mix normally and can transmit or acquire infection; inaccessible individuals have withdrawn from contact in response to perceived risk and do not participate in transmission. Accessibility is a behavioral attribute, not a clinical one: it does not affect the course of infection, only exposure. \Cref{fig:flowchart} summarizes the resulting structure.

\begin{figure}[!ht]
    \centering
    \includegraphics[
        width=0.85\textwidth,
        keepaspectratio
    ]{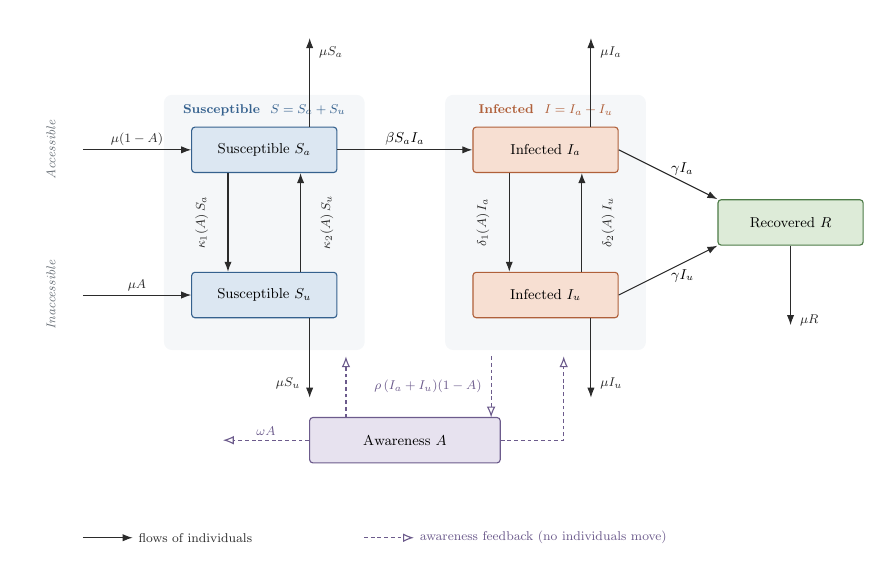}
\caption{\textbf{Structure of the SIRA model.} Susceptible and infectious individuals occupy accessible ($\Sa,\Ia$) or inaccessible ($\Su,\Iu$) behavioral states. Transmission at rate $\beta$ occurs only between accessible susceptible and accessible infectious individuals; recovery ($\gamma$) and natural mortality ($\mu$) act uniformly across accessibility states. Births enter the susceptible class and are split between accessible and inaccessible states in proportion $1-A$ and $A$. Awareness $A$ is generated by prevalence at rate $\rho$, decays through behavioral fatigue at rate $\omega$, and governs the accessibility transition rates $\kappa_1,\kappa_2$ (susceptible) and $\delta_1,\delta_2$ (infectious). Solid arrows denote flows of individuals between compartments; dashed arrows denote the awareness feedback, along which no individuals move. No rate of transmission per contact depends on $A$: awareness acts only on who is available to transmit.}
\label{fig:flowchart}
\end{figure}

\subsection*{Awareness}

Let $A(t)\in[0,1]$ denote a population-level \emph{risk-awareness index:} a dimensionless summary of the prevailing level of perceived epidemic risk, normalized so that $A=0$ denotes the complete absence of concern and $A=1$ its saturation. We emphasize that $A$ is not a compartment and is not the fraction of the population that has actually withdrawn from contact. That fraction is $\Su+\Iu$, which obeys its own equations and does not equal $A$ at any time. What $A$ does is set the rates at which individuals move between accessible and inaccessible states, so that awareness and withdrawal are related dynamically rather than by identity. Awareness accumulates in proportion to prevalence at rate $\rho$, modulated by the factor $1-A$ representing the shrinking margin by which perceived risk can still rise, and is lost at rate $\omega$ through habituation, economic and social pressure, and declining motivation to maintain precautions \citep{Harvey2020,Jorgensen2022,Chen2024}. We refer to $\rho$ as the awareness response and to $\omega$ as behavioral fatigue. The index $A$ is dimensionless, while $\rho$ and $\omega$ carry units of day$^{-1}$, as do all other rates in \Cref{tab:params}.

\subsection*{Awareness-driven accessibility}

Awareness governs movement between behavioral states. Accessible susceptible individuals withdraw at rate $\kappa_1(A)=\kzo A$ and inaccessible susceptible individuals return at rate $\kappa_2(A)=\kzt(1-A)$; the corresponding rates for infectious individuals are $\delta_1(A)=\dzo A$ and $\delta_2(A)=\dzt(1-A)$. Withdrawal rates increase and return rates decrease with awareness, and each is bounded by its baseline maximum. The reciprocals $1/\kzt$ and $1/\dzt$ define baseline behavioral relaxation timescales. At a fixed awareness level $A<1$, the corresponding instantaneous return timescales are $[\kzt(1-A)]^{-1}$ and $[\dzt(1-A)]^{-1}$, respectively. These finite relaxation timescales are the structural feature that distinguishes the present formulation from models in which awareness acts instantaneously by multiplying the transmission rate; Proposition \ref{prop:fastlimit} makes this connection precise.

\subsection*{Transmission and demography}

Transmission occurs at rate $\beta$ through contacts between $\Sa$ and $\Ia$ only. Recovery occurs at rate $\gamma$ irrespective of accessibility. Births and deaths occur at the common per capita rate $\mu$; deaths act uniformly on all classes, and births enter the susceptible class, split between accessible and inaccessible states according to the prevailing awareness level.

\subsection*{Governing equations}

The model is
\begin{align}
\frac{d\Sa}{dt} &= (1-A)\mu-\beta \Sa \Ia+\kappa_2(A)\Su-\kappa_1(A)\Sa-\mu \Sa, \label{eq:1}\\
\frac{d\Su}{dt} &= A\mu+\kappa_1(A)\Sa-\kappa_2(A)\Su-\mu \Su, \label{eq:2}\\
\frac{d\Ia}{dt} &= \beta \Sa \Ia-\delta_1(A)\Ia+\delta_2(A)\Iu-(\gamma+\mu)\Ia, \label{eq:3}\\
\frac{d\Iu}{dt} &= \delta_1(A)\Ia-\delta_2(A)\Iu-(\gamma+\mu)\Iu, \label{eq:4}\\
\frac{dR}{dt} &= \gamma(\Ia+\Iu)-\mu R, \label{eq:5}\\
\frac{dA}{dt} &= \rho(\Ia+\Iu)(1-A)-\omega A. \label{eq:6}
\end{align}
\subsection*{Modeling assumptions}

Three choices in \eqref{eq:1}--\eqref{eq:6} deserve explicit statement. First, recovered individuals are not assigned an accessibility state. Under the interpretation of $A$ as perceived risk, aware recovered individuals would also withdraw from contact, but because they neither transmit nor acquire infection this has no dynamical consequence, and we adopt it as a deliberate simplification. It would require revision under waning immunity, when recovered individuals re-enter the susceptible pool. Second, births are apportioned $(1-A)$ to $\Sa$ and $A$ to $\Su$, which assigns each newborn cohort the prevailing population-average awareness and keeps it behaviorally consistent with the population it enters. Routing all births to $\Sa$ instead is equally defensible; the two conventions differ only in the term $\mu A^{*}/\Delta(A^{*})$ in \eqref{eq:SaSu}, and therefore agree to leading order in $\mu$. Third, awareness is generated by total prevalence $I=\Ia+\Iu$ rather than by the accessible prevalence $\Ia$ alone. We read $\rho I(1-A)$ as a response to reported case counts, which do not distinguish withdrawn from circulating cases and which if anything capture withdrawn cases more completely, since symptomatic isolation and case ascertainment tend to coincide. The alternative $\rho\Ia(1-A)$ would weaken the feedback at equilibrium by exactly the factor $\Ia^{*}/I^{*}=1/\Phi(A^{*})<1$, with $\Phi$ defined in \eqref{eq:Phi}, and would therefore reduce, but not remove, the endemic effects reported below; the consequence of this convention for the stability results of \Cref{sec:hopf} is examined there.

\subsection*{Reduction to the classical SIR model}

The no-awareness limit is $\rho=0$, and the reduction it produces is asymptotic rather than exact. If $\rho=0$ then \eqref{eq:6} gives $A(t)=A(0)e^{-\omega t}\to0$, so the withdrawal rates $\kappa_1(A)=\kzo A$ and $\delta_1(A)=\dzo A$ vanish in the limit, while the return rates remain strictly positive there, $\kappa_2(0)=\kzt>0$ and $\delta_2(0)=\dzt>0$. The inaccessible classes therefore drain and the SIR model with demography is recovered \emph{asymptotically} in $(S,I,R)$. The reduction is \emph{exact for every $t\ge0$} only under the additional initial conditions $A(0)=\Su(0)=\Iu(0)=0$, in which case $A\equiv\Su\equiv\Iu\equiv0$ and $(S,I,R)$ solves the classical system identically. This distinction matters because the $\rho=0$ limit is the no-awareness control used in \Cref{sec:sircontrol}.

Setting the four accessibility rates to zero instead does \emph{not} recover the SIR model when $\rho>0$. In that case \eqref{eq:2} gives $\Su^{*}=A^{*}>0$ at equilibrium, so a positive standing fraction of susceptibles remains inaccessible through the birth apportionment alone and $\Sa^{*}<S^{*}$; only the sequestration channel is switched off, since $\Iu\to0$ at rate $\gamma+\mu$. Accordingly, the SIR benchmark used throughout is defined by $\rho=0$.

Parameter values are listed in \Cref{tab:params} and follow the ranges used in the behavioral epidemic literature for an acute respiratory infection with a mean infectious period of five days.

\begin{table}[htbp]
\centering
\caption{Model parameters, baseline values, and ranges explored. All rates are per capita and expressed in day$^{-1}$; the awareness index $A$ is dimensionless.}
\label{tab:params}
\begin{tabular}{@{}llll@{}}
\toprule
\textbf{Symbol} & \textbf{Description} & \textbf{Baseline / range} & \textbf{Unit} \\
\midrule
$\beta$   & Transmission rate                       & 0.6 / 0.2--2      & day$^{-1}$ \\
$\gamma$  & Recovery rate                           & 0.2               & day$^{-1}$ \\
$\mu$     & Birth and death rate                    & 0.02              & day$^{-1}$ \\
$\kzo$    & Susceptible withdrawal rate             & 0.4 / 0.01--0.8   & day$^{-1}$ \\
$\kzt$    & Susceptible return rate                 & 0.4 / 0.01--0.8   & day$^{-1}$ \\
$\dzo$    & Infectious withdrawal rate              & 0.2 / 0.01--0.6   & day$^{-1}$ \\
$\dzt$    & Infectious return rate                  & 0.2 / 0.01--0.6   & day$^{-1}$ \\
$\rho$    & Awareness response rate                 & 0.5 / 0.05--1     & day$^{-1}$ \\
$\omega$  & Behavioral fatigue rate                 & 0.01 / 0.001--0.5 & day$^{-1}$ \\
\bottomrule
\end{tabular}

\vspace{0.6em}
\begin{minipage}{0.95\textwidth}
\footnotesize
\emph{Notes.} The recovery rate $\gamma=0.2$ corresponds to a mean infectious period of five days, typical of acute respiratory infection, and $\beta=0.6$ then gives $\Rzero\approx2.73$, within the range reported for pandemic influenza and for early SARS-CoV-2. The swept range $\beta\in[0.2,2]$ corresponds to $\Rzero\in[0.91,9.09]$ and is therefore \emph{partly subcritical}: no endemic equilibrium exists for $\beta<\gamma+\mu=0.22$ by \Cref{thm:endemic}, and the endemic branches in \Cref{fig:fig5}c,f and \Cref{fig:fig6}b,c,e,f accordingly originate at that value. The four accessibility parameters $\kzo,\kzt,\dzo,\dzt$ are \emph{maximum} rates, not the rates themselves: $\kappa_1(A)=\kzo A$ and $\delta_1(A)=\dzo A$ attain their maxima at $A=1$, while $\kappa_2(A)=\kzt(1-A)$ and $\delta_2(A)=\dzt(1-A)$ attain theirs at $A=0$. Direct estimates of these behavioral rates are not available, so their baselines are chosen to place the withdrawal and return timescales in the range of days to weeks reported for self-initiated distancing \citep{Poletti2009,Perra2011,Bish2010}, and each is swept over more than an order of magnitude. The baseline choices $\kzo=\kzt$ and $\dzo=\dzt$ impose a symmetry between withdrawal and return that is a convenience rather than an empirical fact; \Cref{thm:monotone} shows that the two act in opposite directions on endemic prevalence, so the symmetry is not innocuous, and \Cref{fig:fig3,fig:fig4} vary each independently. See the discussion of $\mu$ below.
\end{minipage}
\end{table}

\subsection*{Interpretation of the demographic rate}

The value $\mu=0.02$~day$^{-1}$ does not correspond to human demography, and we state plainly what role it plays. It is used as an effective replenishment rate for the susceptible pool, chosen so that demographic renewal, behavioral relaxation, and awareness decay act on commensurate timescales and the endemic regime is reached within a computationally convenient horizon. Three consequences should be borne in mind when reading the quantitative results. First, $\mu$ alone fixes the prevalence ceiling: every equilibrium satisfies $I^{*}<I_{\max}=\mu/(\gamma+\mu)=0.0909$, so the reported prevalences are best read relative to that ceiling and to the SIR benchmark $I^{*}_{\mathrm{SIR}}$ rather than as absolute epidemiological predictions. Second, the inter-epidemic period of the damped waves scales as $\mu^{-1/2}$, as it does in the classical SIR model with demography, so under human demography the intervals reported here in days would become intervals of years; the timings should be read as illustrative of the mechanism rather than as forecasts. Third, and importantly, none of the analytical results depends on this choice:
Propositions~\ref{prop:wellposed}, \ref{prop:R0}, and \ref{prop:fastlimit},
Lemma~\ref{lem:monotone}, Corollary~\ref{cor:forward}, and
Theorems~\ref{thm:local_dfe}, \ref{thm:global_dfe}, \ref{thm:endemic},
\ref{thm:persistence}--\ref{thm:lowfatigue}, and \ref{thm:memoryless}
hold for every $\mu>0$.

A further consequence concerns how the parameterization should be rescaled, since the endemic state depends on $\rho$ and $\omega$ only through $\theta=\omega/\rho$: holding $\rho$ and $\omega$ fixed while lowering $\mu$ drives $A^{*}\to0$ and removes the awareness effect, whereas rescaling $\theta$ in proportion to $\mu$ restores it. \Cref{rem:rescale} states this precisely once the required equilibrium relations are available.

\section{Mathematical analysis}
\label{sec:analysis}

Throughout we write $I=\Ia+\Iu$ and
\begin{equation}
\Phi(A)=1+\frac{\delta_1(A)}{\delta_2(A)+\gamma+\mu}\;\ge\;1 ,
\label{eq:Phi}
\end{equation}
which is the factor by which awareness-driven sequestration inflates the mean time an infection spends outside the transmitting pool relative to a fully accessible infection. We also set $\theta=\omega/\rho$, the fatigue-to-response ratio.

\subsection{Well-posedness and invariant region}

\begin{proposition}
\label{prop:wellposed}
The set
\begin{equation}
\Gam=\Big\{(\Sa,\Su,\Ia,\Iu,R,A)\in\eqR^{6}_{+}:\ \Sa+\Su+\Ia+\Iu+R=1,\ 0\le A\le1\Big\}
\label{eq:region}
\end{equation}
is compact and positively invariant under \eqref{eq:1}--\eqref{eq:6}, and every initial condition in $\Gam$ generates a unique solution defined for all $t\ge0$.
\end{proposition}

\begin{proof}
On each face of the nonnegative orthant the vector field points inward: at $\Sa=0$, $d\Sa/dt=(1-A)\mu+\kappa_2(A)\Su\ge0$; at $\Su=0$, $d\Su/dt=A\mu+\kappa_1(A)\Sa\ge0$; at $\Ia=0$, $d\Ia/dt=\delta_2(A)\Iu\ge0$; at $\Iu=0$, $d\Iu/dt=\delta_1(A)\Ia\ge0$; and at $R=0$, $dR/dt=\gamma I\ge0$. At $A=0$, $dA/dt=\rho I\ge0$ and at $A=1$, $dA/dt=-\omega<0$, so $[0,1]$ is invariant for $A$. Summing \eqref{eq:1}--\eqref{eq:5}, the accessibility exchange terms cancel pairwise and $N=\Sa+\Su+\Ia+\Iu+R$ obeys $dN/dt=\mu(1-N)$, so $N\equiv1$ is invariant. The right-hand sides are polynomial and hence locally Lipschitz on the compact set $\Gam$, so existence, uniqueness, and global forward continuation follow from the Picard--Lindel\"of theorem.
\end{proof}

Biologically, the awareness feedback redistributes individuals between behavioral states without creating or destroying them, so the demographic balance is untouched by behavior.

\subsection{Disease-free equilibrium and threshold neutrality}
\label{sec:R0}

Setting $I=0$ in \eqref{eq:6} forces $A\to0$, and \eqref{eq:1}--\eqref{eq:2} then give a unique disease-free equilibrium
\begin{equation}
E_0=(1,0,0,0,0,0).
\end{equation}
In the absence of infection there is nothing to be aware of, awareness vanishes, and the population is fully accessible. This is the structural reason that behavior cannot influence invasion in this framework.

We take $(\Ia,\Iu)$ as the infected subsystem and exclude $A$ from it. Awareness is driven by infection and vanishes at $E_0$, so a reader may reasonably ask whether it belongs among the infected variables; it does not, because the generation term $\rho I(1-A)$ creates perceived risk rather than new infections, and the next-generation construction requires $\mathcal{F}$ to contain new infections only. The choice is also immaterial to the result: placing $A$ in the infected subsystem with its generation term assigned to $\mathcal{V}$ leaves the $(\Ia,\Iu)$ block of $V$ unchanged and yields the same spectral radius.

\begin{proposition}
\label{prop:R0}
With $\Ia,\Iu$ as infected compartments, the next-generation construction \citep{vandenDriessche2002} gives, with $\varrho(\cdot)$ denoting the spectral radius,
\begin{equation}
F=\begin{pmatrix}\beta & 0\\ 0 & 0\end{pmatrix},\qquad
V=\begin{pmatrix}\gamma+\mu & -\dzt\\ 0 & \dzt+\gamma+\mu\end{pmatrix},\qquad
\Rzero=\varrho\big(FV^{-1}\big)=\frac{\beta}{\gamma+\mu}.
\label{eq:R0}
\end{equation}
The zero in the $(2,1)$ entry of $V$ is not an omission: it arises because $\delta_1(0)=0$, so at the disease-free state there is no flow from $\Ia$ into $\Iu$. Consequently $\Rzero$ is independent of $\kzo,\kzt,\dzo,\dzt,\rho$ and $\omega$.
\end{proposition}

The cancellation in \eqref{eq:R0} is not accidental. At $E_0$ we have $\delta_1(0)=0$, so a newly introduced infection cannot be sequestered; the only accessibility transition present is the return flow $\dzt$, which enters $V$ but not $F$. Because $\delta_1(0)=0$, no new infection enters $\Iu$, so $V$ is upper triangular at $E_0$ and the $(1,1)$ entry alone determines $\varrho(FV^{-1})$; the return flow therefore drops out of the threshold. Awareness is a lagging indicator of prevalence: at the instant of invasion, $I=0$ implies $A=0$, so no behavioral machinery is active. Awareness-driven accessibility therefore cannot prevent an epidemic from starting; it can only shape what happens afterwards. Because $\Rzero$ depends on three parameters only, its normalized forward sensitivity indices are $\Upsilon^{\Rzero}_{\beta}=1$, $\Upsilon^{\Rzero}_{\gamma}=-\gamma/(\gamma+\mu)$ and $\Upsilon^{\Rzero}_{\mu}=-\mu/(\gamma+\mu)$, with all behavioral indices identically zero.

Threshold neutrality holds for every parameter value in this model, but it is a consequence of the specific endogenous-awareness construction rather than a general property of awareness models, and we do not claim otherwise. It requires both $\delta_1(A)=\dzo A$, so that withdrawal of infectious individuals is proportional to awareness, and $A^{*}=0$ at the disease-free equilibrium, so that no awareness is present at the moment of introduction. Either feature can fail. With any exogenous source of awareness, for instance media coverage, imported information, or memory of an earlier outbreak, modeled as
\begin{equation}
\frac{dA}{dt}=\rho I(1-A)+\eta(1-A)-\omega A ,
\label{eq:exogenous}
\end{equation}
the disease-free awareness level becomes $A^{\dagger}=\eta/(\eta+\omega)>0$. The disease-free equilibrium of \eqref{eq:exogenous} still has $S^{\dagger}=1$, now split as $\Sa^{\dagger}=\big[\kappa_2(A^{\dagger})+\mu-\mu A^{\dagger}\big]/\Delta(A^{\dagger})$, and the next-generation construction applied at that state gives
\begin{equation}
\Rzero^{\eta}
=\frac{\beta\,\Sa^{\dagger}}{(\gamma+\mu)\,\Phi(A^{\dagger})},
\qquad
A^{\dagger}=\frac{\eta}{\eta+\omega},
\qquad
\Sa^{\dagger}=\frac{\kappa_2(A^{\dagger})+\mu-\mu A^{\dagger}}{\Delta(A^{\dagger})} .
\label{eq:R0eta}
\end{equation}
Both behavioral channels now enter the threshold: withdrawal reduces $\Sa^{\dagger}$ below unity and sequestration inflates $\Phi(A^{\dagger})$ above unity, so $\Rzero^{\eta}$ depends on $\kzo,\kzt,\dzo,\dzt$ and on $\eta$ and $\omega$. Setting $\eta=0$ gives $A^{\dagger}=0$, $\Sa^{\dagger}=1$, $\Phi(0)=1$ and recovers $\Rzero=\beta/(\gamma+\mu)$. Expression \eqref{eq:R0eta} is the endemic balance identity \eqref{eq:balance} evaluated at the disease-free awareness level, which is the structural reason the two coincide: in both cases the effective reproduction number is the product of the accessible susceptible fraction and the awareness-inflated mean infectious residence time. \citet{Funk2009} draw the same distinction between awareness that spreads only from infection and awareness that spreads independently of it. Threshold neutrality should therefore be read as a statement about purely prevalence-generated awareness, and as identifying the condition under which behavior cannot move the invasion threshold, rather than as a universal structural theorem.

\begin{theorem}
\label{thm:local_dfe}
$E_0$ is locally asymptotically stable if $\Rzero<1$ and unstable if $\Rzero>1$.
\end{theorem}

\begin{proof}
Eliminating $\Sa$ through the conservation constraint and linearizing in $(\Su,\Ia,\Iu,R,A)$ at $E_0$, where $\kappa_1(0)=\delta_1(0)=0$, $\kappa_2(0)=\kzt$ and $\delta_2(0)=\dzt$, gives a Jacobian that is block triangular with respect to the splitting $\{\Su\}\oplus\{\Ia,\Iu,R,A\}$, since the $(\Ia,\Iu,R,A)$ block does not depend on $\Su$. Its eigenvalues are
\begin{equation}
\lambda_1=-(\kzt+\mu),\quad \lambda_2=\beta-(\gamma+\mu),\quad \lambda_3=-(\dzt+\gamma+\mu),\quad \lambda_4=-\mu,\quad \lambda_5=-\omega .
\label{eq:eigs}
\end{equation}
All but $\lambda_2=(\gamma+\mu)(\Rzero-1)$ are strictly negative for admissible parameters.
\end{proof}

The spectrum \eqref{eq:eigs} is real, so no oscillatory instability can arise at $E_0$: behavioral parameters enter only through $\lambda_1,\lambda_3,\lambda_5$, that is, they set how fast the behavioral subsystem relaxes but never whether the disease-free state is stable. This is essentially automatic for a disease-free equilibrium of this type and carries little information. The substantive question is whether a complex pair can cross the imaginary axis at the \emph{endemic} equilibrium, which is taken up in \Cref{sec:hopf}.

\begin{theorem}
\label{thm:global_dfe}
If $\Rzero<1$, $E_0$ is globally asymptotically stable on $\Gam$.
\end{theorem}

\begin{proof}
Adding \eqref{eq:3} and \eqref{eq:4}, the accessibility exchange terms cancel exactly and
\begin{equation}
\frac{dI}{dt}=\beta \Sa\Ia-(\gamma+\mu)I\le(\gamma+\mu)(\Rzero-1)I,
\label{eq:Idot}
\end{equation}
since $\Sa\le1$ and $\Ia\le I$ on $\Gam$. The comparison lemma gives $I(t)\le I(0)e^{(\gamma+\mu)(\Rzero-1)t}\to0$. Then \eqref{eq:6} is a linear equation for $A$ with exponentially vanishing forcing and decay rate at least $\omega$, so $A(t)\to0$, and \eqref{eq:5} gives $R(t)\to0$. Finally \eqref{eq:2} becomes the scalar equation $d\Su/dt=-\big(\kappa_2(A(t))+\mu\big)\Su+\varepsilon(t)$ with $\varepsilon(t)=A(t)\mu+\kappa_1(A(t))\Sa\to0$, whose time-varying decay coefficient satisfies $\kappa_2(A(t))+\mu\ge\mu>0$ uniformly in $t$ and converges to $\kzt+\mu$. The system is therefore asymptotically autonomous with limit equation $d\Su/dt=-(\kzt+\mu)\Su$, whose equilibrium $\Su=0$ is globally asymptotically stable; since the forward orbit has compact closure in $\Gam$, the convergence theory for asymptotically autonomous semiflows \citep{Thieme1992,SmithThieme2011} gives $\Su(t)\to0$ and hence $\Sa(t)\to1$.
\end{proof}

The exact cancellation in \eqref{eq:Idot} states that accessibility transitions relabel infectious individuals without changing their number. This is what makes the elimination threshold behavior-independent even globally, and it is the precise sense in which awareness restructures transmission rather than removing infection.

\subsection{Reduction of the endemic problem to a scalar equation}

At any equilibrium with $I^{*}>0$, setting the right-hand sides of \eqref{eq:1}--\eqref{eq:6} to zero gives
\begin{align}
A^{*}&=\frac{\rho I^{*}}{\rho I^{*}+\omega},\qquad
R^{*}=\frac{\gamma}{\mu}I^{*},\qquad
S^{*}=\Sa^{*}+\Su^{*}=1-\frac{\gamma+\mu}{\mu}I^{*},
\label{eq:AR}\\[2pt]
\Ia^{*}&=\frac{\delta_2(A^{*})+\gamma+\mu}{D(A^{*})}I^{*},\qquad
\Iu^{*}=\frac{\delta_1(A^{*})}{D(A^{*})}I^{*},\qquad D(A)=\delta_1(A)+\delta_2(A)+\gamma+\mu,
\label{eq:IaIu}\\[2pt]
\Sa^{*}&=\frac{(\kappa_2(A^{*})+\mu)S^{*}-\mu A^{*}}{\Delta(A^{*})},\qquad
\Su^{*}=\frac{\kappa_1(A^{*})S^{*}+\mu A^{*}}{\Delta(A^{*})},
\label{eq:SaSu}
\end{align}
with $\Delta(A)=\kappa_1(A)+\kappa_2(A)+\mu\ge\min(\kzo,\kzt)+\mu>0$. Note that $\Sa^{*}+\Su^{*}=S^{*}$ identically. Every coordinate is thus an explicit function of the single scalar $I^{*}$.

Adding \eqref{eq:1} and \eqref{eq:2} at equilibrium, all accessibility terms cancel and $\beta\Sa^{*}\Ia^{*}=\mu(1-S^{*})=(\gamma+\mu)I^{*}$. Substituting $\Ia^{*}$ from \eqref{eq:IaIu} and dividing by $I^{*}>0$ yields the \emph{endemic balance identity}
\begin{equation}
\beta\,\Sa^{*}=(\gamma+\mu)\,\Phi(A^{*}).
\label{eq:balance}
\end{equation}
Equation \eqref{eq:balance} is the statement that the effective reproduction number equals one at equilibrium, and it separates the two channels through which awareness acts. On the left, withdrawal reduces the accessible susceptible pool below the total susceptible pool. On the right, sequestration inflates the effective removal time of an infection by the factor $\Phi(A^{*})>1$. Setting $A^{*}=0$ gives $\Phi=1$ and recovers the classical condition $\beta S^{*}=\gamma+\mu$.

Define, for $I\in[0,I_{\max}]$ with $I_{\max}=\mu/(\gamma+\mu)$ the value at which $S^{*}=0$,
\begin{equation}
\mathcal{A}(I)=\frac{\rho I}{\rho I+\omega},\qquad
\Sfun(I)=\frac{\big(\kappa_2(\mathcal{A}(I))+\mu\big)\Big(1-\tfrac{\gamma+\mu}{\mu}I\Big)-\mu \mathcal{A}(I)}{\Delta(\mathcal{A}(I))},
\label{eq:Sfun}
\end{equation}
\begin{equation}
G(I)=\beta\,\Sfun(I)-(\gamma+\mu)\,\Phi(\mathcal{A}(I)).
\label{eq:G}
\end{equation}
Endemic equilibria in $\Gam$ correspond exactly to roots of $G$ at which $\Sfun\ge0$. Let $\hat I\in(0,I_{\max}]$ be the first zero of $\Sfun$, or $\hat I=I_{\max}$ if $\Sfun$ has no zero there.

\begin{lemma}[Strict monotonicity]
\label{lem:monotone}
The function $\Sfun$ is strictly decreasing on $[0,\hat I]$ and,
for fixed $I$, strictly decreasing with respect to $A$. Consequently,
$G$ is strictly decreasing on $[0,\hat I]$, and
\[
G(I)<0, \qquad I\in(\hat I,I_{\max}].
\]
\end{lemma}

\begin{proof}
Let
\[
a=\mathcal{A}(I),\qquad
s=1-\frac{\gamma+\mu}{\mu}I,\qquad
N=\bigl(\kzt(1-a)+\mu\bigr)s-\mu a,
\]
and
\[
\Delta=\kzt+\mu-(\kzt-\kzo)a.
\]
Then $\Sfun=N/\Delta$, with $N\ge0$ on $[0,\hat I]$ and
$\Delta>0$. Differentiation gives
\begin{equation}
\Delta^{2}\frac{d\Sfun}{dI}
=
\bigl(\kzt(1-a)+\mu\bigr)s'(I)\Delta
-\mathcal{A}'(I)B,
\qquad
B=(\kzt s+\mu)\Delta+N(\kzo-\kzt),
\label{eq:dS}
\end{equation}
where
\[
s'(I)=-\frac{\gamma+\mu}{\mu}<0,
\qquad
\mathcal{A}'(I)=\frac{\rho\omega}{(\rho I+\omega)^2}>0.
\]

If $\kzo\ge\kzt$, then
\[
B\ge(\kzt s+\mu)\Delta>0.
\]
If $\kzo<\kzt$, set $c=\kzt-\kzo>0$. Since
\[
N\le\bigl(\kzt(1-a)+\mu\bigr)s,
\]
we obtain
\[
\begin{aligned}
B
&\ge
\mu\Delta
+s\Bigl[
\kzt\Delta
-c\bigl(\kzt(1-a)+\mu\bigr)
\Bigr] \\
&=
\mu\Delta+s\,\kzo(\kzt+\mu)>0.
\end{aligned}
\]
Hence $d\Sfun/dI<0$. Moreover,
\[
\frac{\partial\Sfun}{\partial a}
=-\frac{B}{\Delta^2}<0.
\]

Since $\mathcal{A}(I)$ and $\Phi(A)$ are strictly increasing,
both $\beta\Sfun(I)$ and
$-(\gamma+\mu)\Phi(\mathcal{A}(I))$ are strictly decreasing.
Therefore $G$ is strictly decreasing on $[0,\hat I]$.
Finally, $\Sfun(I)\le0$ on $(\hat I,I_{\max}]$, whereas
$\Phi(\mathcal{A}(I))>0$, and hence
\[
G(I)<0,\qquad I\in(\hat I,I_{\max}].
\]
\end{proof}

\begin{theorem}[Existence and uniqueness of the endemic equilibrium]
\label{thm:endemic}
System \eqref{eq:1}--\eqref{eq:6} has no endemic equilibrium in
$\Gam$ when $\Rzero\le1$ and exactly one endemic equilibrium,
denoted $E^{*}$, when $\Rzero>1$. Moreover,
\[
I^{*}\in(0,I_{\max}),\qquad
I_{\max}=\frac{\mu}{\gamma+\mu},
\qquad
I^{*}\to0
\quad\text{as}\quad
\Rzero\downarrow1.
\]
In particular, no subthreshold endemic equilibrium exists.
\end{theorem}

\begin{proof}
Since
\[
\Sfun(0)=\Phi(\mathcal{A}(0))=1,
\]
we have
\[
G(0)
=\beta-(\gamma+\mu)
=(\gamma+\mu)(\Rzero-1).
\]

If $\Rzero\le1$, then $G(0)\le0$. By Lemma
\ref{lem:monotone},
\[
G(I)<0
\qquad\text{for every}\qquad
I\in(0,I_{\max}],
\]
so no endemic equilibrium exists.

If $\Rzero>1$, then $G(0)>0$, whereas
\[
G(\hat I)
=-(\gamma+\mu)\Phi(\mathcal{A}(\hat I))<0.
\]
Continuity and strict monotonicity of $G$ therefore yield a unique
root
\[
I^{*}\in(0,\hat I)\subset(0,I_{\max}).
\]

At this root,
\[
\beta\Sfun(I^{*})
=(\gamma+\mu)\Phi(\mathcal{A}(I^{*}))>0,
\]
so
\[
\Sa^{*}=\Sfun(I^{*})>0.
\]
Also,
\[
S^{*}
=1-\frac{\gamma+\mu}{\mu}I^{*}>0,
\]
and \eqref{eq:SaSu} gives $\Su^{*}\ge0$. The remaining equilibrium
relations \eqref{eq:AR}--\eqref{eq:IaIu} yield
\[
\Ia^{*},\Iu^{*},R^{*}\ge0,
\qquad
A^{*}\in(0,1).
\]
Thus the unique root determines a unique admissible endemic
equilibrium $E^{*}\in\Gam$.

Finally, as $\Rzero\downarrow1$, we have $G(0)\to0$. At
$\Rzero=1$, $G(0)=0$, while Lemma \ref{lem:monotone} gives
\[
G(I)<0
\qquad\text{for every}\qquad I>0.
\]
Therefore, by continuity and uniqueness of the endemic root,
\[
I^{*}\to0
\qquad\text{as}\qquad
\Rzero\downarrow1.
\]
\end{proof}

\begin{corollary}
\label{cor:forward}
The bifurcation at $\Rzero=1$ is forward and transcritical for all admissible parameter values. In particular there is no subthreshold endemic state and no bistability, and $E^{*}$ is locally asymptotically stable for $\Rzero$ in a right neighborhood of $1$.
\end{corollary}

\begin{proof}
By \Cref{thm:endemic} the endemic branch exists for $\Rzero>1$ and for no smaller value of $\Rzero$, and $I^{*}\to0$ as $\Rzero\downarrow1$, so the branch emerges continuously from $E_0$ on the supercritical side only; this already excludes a backward branch and any subthreshold endemic state, for all admissible parameters and not merely locally. To identify the crossing as transcritical and nondegenerate, note that $G(I;\beta)=\beta\Sfun(I)-(\gamma+\mu)\Phi(\mathcal{A}(I))$ satisfies $\partial G/\partial\beta=\Sfun(I)>0$ and, by Lemma \ref{lem:monotone}, $\partial G/\partial I<0$. Implicit differentiation of $G(I^{*};\beta)=0$ therefore gives
\begin{equation}
\frac{dI^{*}}{d\beta}=-\frac{\Sfun(I^{*})}{G'(I^{*})}>0 ,
\qquad
\lim_{\beta\downarrow\gamma+\mu}\frac{dI^{*}}{d\beta}=\frac{\Sfun(0)}{|G'(0)|}=\frac{1}{|G'(0)|}\in(0,\infty),
\end{equation}
so the branch leaves the threshold with finite positive slope and the crossing is transversal. At $\beta^{*}=\gamma+\mu$ the Jacobian $J(E_0)$ has the simple zero eigenvalue $\lambda_2$ of \eqref{eq:eigs}, while the remaining eigenvalues $-(\kzt+\mu)$, $-(\dzt+\gamma+\mu)$, $-\mu$ and $-\omega$ are strictly negative and bounded away from the imaginary axis. The center manifold is therefore one dimensional, and the standard exchange of stability for a transversal transcritical bifurcation with a simple zero eigenvalue \citep{CastilloChavez2004} gives local asymptotic stability of $E^{*}$ for $\Rzero$ in a right neighborhood of $1$. In the notation of \citet{CastilloChavez2004}, the forward direction established above corresponds to the statement that the bifurcation coefficients satisfy $a<0$ and $b>0$; we obtain it from the global monotonicity of $G$, which holds over the whole admissible parameter range, rather than from an expansion valid only near threshold.
\end{proof}

The epidemiological content of Corollary \ref{cor:forward} is that reducing $\Rzero$ below one always eliminates infection in this model. Behavior can make elimination easier or harder to approach, but it cannot create the subthreshold persistence that backward bifurcations produce in models where protective behavior is imperfect in a different sense.

\subsection{Uniform persistence of infection and of awareness}

\begin{theorem}
\label{thm:persistence}
Let $\Rzero>1$. There exists $\vartheta>0$, independent of initial conditions, such that every solution with $\Ia(0)+\Iu(0)>0$ satisfies
\begin{equation}
\liminf_{t\to\infty}I(t)\ \ge\ \vartheta,
\qquad
\liminf_{t\to\infty}A(t)\ \ge\ \frac{\rho\vartheta}{\rho\vartheta+\omega}\;>\;0 .
\end{equation}
\end{theorem}

\begin{proof}
$\Gam$ is compact and positively invariant, so the semiflow is point dissipative. The set $X_0=\{x\in\Gam:\Ia=\Iu=0\}$ is closed and invariant, and on $X_0$ we have $dA/dt=-\omega A$, $dR/dt=-\mu R$, hence $A,R\to0$, and then $\Su\to0$ and $\Sa\to1$; so $\{E_0\}$ is the maximal invariant set in $X_0$, and it is isolated and acyclic because it is hyperbolic for $\Rzero\ne1$. It remains to show $W^{s}(E_0)\cap(\Gam\setminus X_0)=\emptyset$. If $\Ia(0)+\Iu(0)>0$ then $\Ia(t)>0$ for all $t>0$, since $\Ia=0$ with $\Iu>0$ implies $d\Ia/dt=\delta_2(A)\Iu>0$ once $A<1$, which holds immediately because $dA/dt=-\omega<0$ at $A=1$. Suppose such a solution converged to $E_0$. Then $\Sa\to1$ and $A\to0$, so for $\varepsilon\in(0,\beta-(\gamma+\mu))$ there is $T$ with $\beta\Sa(t)-\dzo A(t)-(\gamma+\mu)\ge\beta-(\gamma+\mu)-\varepsilon>0$ for $t\ge T$. Since $d\Ia/dt\ge(\beta\Sa-\delta_1(A)-(\gamma+\mu))\Ia$, we get $\Ia(t)\ge \Ia(T)e^{(\beta-\gamma-\mu-\varepsilon)(t-T)}\to\infty$, contradicting $\Ia\to0$. Uniform persistence of $I$ now follows from the standard acyclicity criterion \citep[Thm.~4.6]{Thieme1993}; see also \citet[Ch.~8]{SmithThieme2011}. Given $\liminf I\ge\vartheta$, for any $\varepsilon>0$ and large $t$, \eqref{eq:6} gives $dA/dt\ge\rho(\vartheta-\varepsilon)(1-A)-\omega A$, whose unique equilibrium $\rho(\vartheta-\varepsilon)/(\rho(\vartheta-\varepsilon)+\omega)$ is globally attracting from below; letting $\varepsilon\to0$ gives the stated bound.
\end{proof}

\Cref{thm:persistence} makes the coexistence in the title a theorem rather than an observation, and it does so without reference to the identity of the attractor. Infection cannot be driven out by behavior, and awareness cannot be sustained without infection: the two persist or vanish together. The mechanism that suppresses transmission is itself powered by the transmission it suppresses, so protective behavior in this framework is intrinsically self-limiting.

\subsection{How much burden does awareness remove?}

Let $I^{*}_{\mathrm{SIR}}=\frac{\mu}{\gamma+\mu}\big(1-\Rzero^{-1}\big)$ denote the endemic prevalence of the classical SIR model with the same $\beta,\gamma,\mu$, that is, of the $\rho=0$ control identified above.

\begin{theorem}[Exact decomposition of burden reduction]
\label{thm:decomp}
Let $\Rzero>1$. Then
\begin{equation}
I^{*}_{\mathrm{SIR}}-I^{*}
=\frac{\mu}{\gamma+\mu}\left[\underbrace{\Su^{*}}_{\text{susceptible withdrawal}}
+\underbrace{\frac{\Phi(A^{*})-1}{\Rzero}}_{\text{infectious sequestration}}\right],
\label{eq:decomp}
\end{equation}
where $\dfrac{\Phi(A^{*})-1}{\Rzero}=\dfrac{\gamma+\mu}{\beta}\cdot\dfrac{\dzo A^{*}}{\dzt(1-A^{*})+\gamma+\mu}$. Both terms are strictly positive, since $A^{*}\in(0,1)$ whenever $\Rzero>1$, and both tend to zero as $A^{*}\to0$. In particular $I^{*}<I^{*}_{\mathrm{SIR}}$ for every admissible parameter set with $\Rzero>1$.
\end{theorem}

\begin{proof}
From \eqref{eq:AR}, $I^{*}=\frac{\mu}{\gamma+\mu}(1-S^{*})$ and $I^{*}_{\mathrm{SIR}}=\frac{\mu}{\gamma+\mu}(1-\Rzero^{-1})$, so $I^{*}_{\mathrm{SIR}}-I^{*}=\frac{\mu}{\gamma+\mu}(S^{*}-\Rzero^{-1})$. The balance identity \eqref{eq:balance} gives $\Sa^{*}=\Phi(A^{*})/\Rzero$, and $S^{*}=\Sa^{*}+\Su^{*}$, hence $S^{*}-\Rzero^{-1}=(\Phi(A^{*})-1)/\Rzero+\Su^{*}$. Since $\Su^{*}=(\kzo A^{*}S^{*}+\mu A^{*})/\Delta(A^{*})$ and $\Phi(A^{*})-1=\dzo A^{*}/(\dzt(1-A^{*})+\gamma+\mu)$, both are proportional to $A^{*}$ and hence strictly positive for $A^{*}>0$ and vanishing in the limit $A^{*}\to0$.
\end{proof}

Equation \eqref{eq:decomp} is an exact accounting of what awareness buys, and it is additive: the two behavioral channels do not interact at equilibrium. The first channel is the standing stock of withdrawn susceptibles, which reduces the pool that can be infected. Part of that stock is contributed by the birth apportionment rather than by withdrawal, through the term $\mu A^{*}/\Delta(A^{*})$ of \eqref{eq:SaSu}, which is present even when $\kzo=0$. At baseline this term equals $0.0208$ against $\Su^{*}=0.366$, so it accounts for $5.7\%$ of the withdrawal channel and $4.5\%$ of the total reduction; recomputing the equilibrium with all births routed to $\Sa$ gives $I^{*}=0.0165$ and shifts the decomposition from $79\%/21\%$ to $78\%/22\%$. The channel is therefore dominated by genuine withdrawal under either convention. The second is the fraction of infectious time spent outside the transmitting pool, which reduces the pool that can infect. The decomposition also explains why the two channels respond oppositely to withdrawal and return rates, as the next result confirms.

\begin{theorem}[Comparative statics of the endemic state]
\label{thm:monotone}
Let $\Rzero>1$ and let all parameters be admissible with $A^{*}>0$. Then
\begin{equation}
\frac{\partial I^{*}}{\partial\beta}>0,\quad
\frac{\partial I^{*}}{\partial\kzo}<0,\quad
\frac{\partial I^{*}}{\partial\kzt}>0,\quad
\frac{\partial I^{*}}{\partial\dzo}<0,\quad
\frac{\partial I^{*}}{\partial\dzt}>0,\quad
\frac{\partial I^{*}}{\partial\rho}<0,\quad
\frac{\partial I^{*}}{\partial\omega}>0 .
\end{equation}
\end{theorem}

\begin{proof}
By Lemma \ref{lem:monotone}, $\partial G/\partial I<0$ at $I^{*}$, so the implicit function theorem gives $\operatorname{sign}(\partial I^{*}/\partial p)=\operatorname{sign}(\partial G/\partial p)$ for each parameter $p$, with $I$ held fixed. Then: $\partial G/\partial\beta=\Sfun(I^{*})>0$; $\partial\Sfun/\partial\kzo=-N\mathcal A/\Delta^{2}<0$; $\partial\Sfun/\partial\kzt=(1-\mathcal A)(s\Delta-N)/\Delta^{2}=(1-\mathcal A)\Su^{*}/\Delta>0$; $\partial\Phi/\partial\dzo=A^{*}/(\delta_2(A^{*})+\gamma+\mu)>0$ so $\partial G/\partial\dzo<0$; $\partial\Phi/\partial\dzt=-\delta_1(A^{*})(1-A^{*})/(\delta_2(A^{*})+\gamma+\mu)^{2}<0$ so $\partial G/\partial\dzt>0$. Finally $\partial G/\partial A=\beta\,\partial\Sfun/\partial A-(\gamma+\mu)\Phi'(A)<0$ by Lemma \ref{lem:monotone}, while $\partial\mathcal A/\partial\rho>0$ and $\partial\mathcal A/\partial\omega<0$, which gives the last two signs.
\end{proof}

\Cref{thm:monotone} states that the sign of every behavioral effect on endemic prevalence is fixed, independently of parameter values: nothing in this model can reverse the direction of a behavioral intervention. Faster withdrawal always helps and faster return always hurts, on both the susceptible and the infectious side. This is a stronger and more useful statement than a numerical sensitivity ranking, because it holds over the whole admissible parameter space, and it means that the numerical results of \Cref{sec:results} are representative rather than parameter-specific.

\subsection{A single behavioral control parameter}

\begin{theorem}[Reduction to the fatigue-to-response ratio]
\label{thm:theta}
Every coordinate of $E^{*}$ depends on $\rho$ and $\omega$ only through $\theta=\omega/\rho$. Moreover $A^{*}$ is strictly decreasing and $I^{*}$ strictly increasing in $\theta$.
\end{theorem}

\begin{proof}
Equation \eqref{eq:AR} is equivalent to $I^{*}=\theta A^{*}/(1-A^{*})$. Substituting this into \eqref{eq:Sfun}--\eqref{eq:G} and using $A$ as the unknown gives $H(A;\theta)=\beta\,\widetilde{\Sfun}(A;\theta)-(\gamma+\mu)\Phi(A)$, in which $\rho$ and $\omega$ appear only through $\theta$. Since $I$ is increasing in $A$ for fixed $\theta$, Lemma \ref{lem:monotone} implies $H$ is strictly decreasing in $A$, so its root $A^{*}(\theta)$ is unique; the remaining coordinates follow from \eqref{eq:AR}--\eqref{eq:SaSu}. Increasing $\theta$ at fixed $A$ decreases $s$ and hence $\widetilde{\Sfun}$, so $\partial H/\partial\theta<0$ and $A^{*}$ decreases in $\theta$; monotonicity of $I^{*}$ follows from \Cref{thm:monotone}.
\end{proof}

Behavioral policy therefore has a single equilibrium lever. Doubling the speed at which a population perceives risk is equivalent, for the long-run endemic state, to halving the rate at which it becomes fatigued. This collapse holds only asymptotically: the transient trajectory does depend on $\rho$ and $\omega$ separately, since $\rho$ sets the rate at which awareness tracks prevalence, and so does the stability of $E^{*}$ (\Cref{sec:hopf}).

This reduction also settles how the model should be rescaled when $\mu$ is changed, a question already raised in \S2. Because $I^{*}=\theta A^{*}/(1-A^{*})$, an endemic awareness level of order unity requires $\theta$ to be of the same order as the prevalence scale $I_{\max}=\mu/(\gamma+\mu)$, so reducing $\mu$ at fixed $\rho,\omega$ rescales the behavioral problem along with the epidemiological one.

\begin{remark}[Demographic rescaling]
\label{rem:rescale}
Write $y=I^{*}/I_{\max}$. By \Cref{thm:theta} together with \eqref{eq:AR}, the endemic problem in the variables $(y,A^{*})$ is
\begin{equation}
\beta\,\Sa^{*}=(\gamma+\mu)\,\Phi(A^{*}),\qquad
y=\frac{\theta}{I_{\max}}\cdot\frac{A^{*}}{1-A^{*}},\qquad
\Sa^{*}=\frac{\big(\kappa_2(A^{*})+\mu\big)(1-y)-\mu A^{*}}{\Delta(A^{*})} ,
\label{eq:rescaled}
\end{equation}
so the demographic rate enters only through the combination $\gamma+\mu$ and through the terms of $\Sa^{*}$ carrying an explicit factor $\mu$, which are of relative size $O\big(\mu/\min(\kzo,\kzt)\big)$. Consequently, if $\theta$ is held in fixed proportion to $I_{\max}$, then $A^{*}$ and the relative burden $I^{*}/I^{*}_{\mathrm{SIR}}$ converge as $\mu\to0$ and are invariant up to $O(\mu/\min(\kzo,\kzt))$.
\end{remark}

Numerically, holding $\theta/I_{\max}$ at its baseline value $0.22$ and leaving all other parameters at their baseline values, we obtain $(A^{*},\,I^{*}/I^{*}_{\mathrm{SIR}})=(0.4372,\,0.2698)$ at $\mu=0.02$, $(0.4626,\,0.2847)$ at $\mu=10^{-3}$, and $(0.4639,\,0.2856)$ at $\mu=3.9\times10^{-5}$, so the invariance predicted by \Cref{rem:rescale} holds to within a few percent across three orders of magnitude in $\mu$. The rescaling argument does not cover stability, since $\rho$ and $\omega$ enter the Jacobian separately (\Cref{sec:hopf}), and we therefore checked that case directly. In the slow-return corner $\kzt=0.01$ with $\theta$ rescaled as above, the leading eigenvalue of $J(E^{*})$ is $+0.00143\pm0.0374i$ at $\mu=0.02$, $+0.00191\pm0.0187i$ at $\mu=5\times10^{-3}$, and $+0.00076\pm0.0087i$ at $\mu=10^{-3}$: the instability of \Cref{sec:hopf} is not an artifact of the inflated demographic rate, although the oscillation frequency scales as $\mu^{1/2}$, in agreement with the $\mu^{-1/2}$ period scaling noted in \Cref{sec:model}. The regime explored throughout this paper is thus one point on a one-parameter family, and the qualitative conclusions transfer to realistic demography provided the fatigue-to-response ratio is rescaled with the prevalence scale.

\begin{theorem}[Low-fatigue limit]
\label{thm:lowfatigue}
Let $\Rzero>1$ and let $A_{\infty}\in(0,1)$ be the unique root of
\begin{equation}
\beta\,\frac{\kappa_2(A)+\mu-\mu A}{\Delta(A)}=(\gamma+\mu)\,\Phi(A).
\label{eq:Ainf}
\end{equation}
Then, as $\theta\to0^{+}$,
\begin{equation}
A^{*}(\theta)\to A_{\infty},
\qquad
I^{*}(\theta)=\theta\,\frac{A_{\infty}}{1-A_{\infty}}+o(\theta).
\end{equation}
\end{theorem}

\begin{proof}
Setting $\theta=0$ in $H$ fixes $s=1$ and yields \eqref{eq:Ainf}, whose left side at $A=0$ equals $\beta$ and at $A=1$ equals $0$, while the right side is at least $\gamma+\mu<\beta$ at $A=0$; strict monotonicity of $H(\cdot;0)$ gives a unique root $A_{\infty}\in(0,1)$. On any compact subinterval of $[0,1)$ containing $A_{\infty}$, $H(\cdot;\theta)\to H(\cdot;0)$ uniformly, so $A^{*}(\theta)\to A_{\infty}$ by continuity of simple roots of a strictly monotone family. Substituting into $I^{*}=\theta A^{*}/(1-A^{*})$ gives the expansion.
\end{proof}

\Cref{thm:lowfatigue} is the counterpart of threshold neutrality, and the two together define the reach of behavior in this framework. Awareness cannot lower the invasion threshold by any amount, yet as fatigue becomes negligible relative to responsiveness it can drive the endemic burden to zero linearly in $\theta$, while awareness itself saturates at a strictly positive plateau $A_{\infty}$ that is set by transmission and accessibility parameters and is independent of $\rho$ and $\omega$. Behavioral fatigue is thus the sole obstruction to asymptotic elimination through behavior, which places the erosion of protective behavior, rather than its initial strength, at the center of long-run epidemic burden. Elimination is approached but never attained: $I^{*}>0$ for every $\theta>0$, in agreement with \Cref{thm:persistence}.

\subsection{The no-awareness control: what awareness adds to the waves}
\label{sec:sircontrol}

The transient trajectories of \Cref{sec:results} display several damped waves before settling, and it would be a mistake to attribute those waves to awareness without a control. The model contains vital dynamics, and the classical SIR model with demography already approaches its endemic equilibrium through damped oscillations. The appropriate control is the $\rho=0$ limit, and because both equilibria are available in closed form the comparison can be made analytically rather than by inspecting simulations.

At $E^{*}_{\mathrm{SIR}}=(\Rzero^{-1},I^{*}_{\mathrm{SIR}})$ the Jacobian of the $(S,I)$ subsystem has trace $-\mu\Rzero$ and determinant $\mu(\Rzero-1)(\gamma+\mu)$, using $\beta I^{*}_{\mathrm{SIR}}=\mu(\Rzero-1)$, so its eigenvalues are
\begin{equation}
\lambda^{\mathrm{SIR}}_{\pm}=\tfrac12\Big[-\mu\Rzero\pm\sqrt{\mu^{2}\Rzero^{2}-4\mu(\Rzero-1)(\gamma+\mu)}\,\Big],
\label{eq:sireigs}
\end{equation}
which form a complex pair whenever $\mu\Rzero^{2}<4(\Rzero-1)(\gamma+\mu)$, as at the baseline parameters. The waves in the SIRA simulations are therefore, at minimum, partly demographic in origin, and the honest question is not whether awareness creates them but what it does to them.

At the baseline values of \Cref{tab:params}, \eqref{eq:sireigs} gives $\lambda^{\mathrm{SIR}}_{\pm}=-0.0273\pm0.0828i$, an inter-epidemic period of $2\pi/0.0828\approx76$ days, a relaxation time of $1/0.0273\approx37$ days, and a damping ratio $|\operatorname{Re}\lambda|/|\operatorname{Im}\lambda|\approx0.329$. Throughout, ``damping ratio'' denotes this quantity, the ratio of the decay rate to the angular frequency; it is used consistently for every model compared here. It differs from the convention $|\operatorname{Re}\lambda|/|\lambda|$ common in linear vibration theory, which would give $0.313$ for the control and $0.141$ for the full model and would not change any comparison below. The corresponding leading pair of $J(E^{*})$ for the full model is $-0.0086\pm0.0601i$ (\Cref{sec:hopf}), giving a period of about $105$ days, a relaxation time of about $116$ days, and a damping ratio of about $0.143$. Awareness-driven accessibility thus lengthens the inter-epidemic period by roughly $38\%$ and reduces the damping ratio by more than half, so the ratio of relaxation time to period rises from about $0.48$ to about $1.11$: the control admits barely one discernible wave, whereas the SIRA model admits several. This is what produces the multi-wave structure of \Cref{fig:fig2,fig:fig3}, and it is a restructuring and prolongation of oscillations already present under demographic renewal rather than their creation.

Two qualifications complete the picture. First, the effect on burden is not shared with the control at all: the control has $I^{*}=I^{*}_{\mathrm{SIR}}$ by definition, whereas \Cref{thm:decomp} shows that awareness strictly reduces prevalence whenever $A^{*}>0$. Second, and more sharply, the damped waves of the control cannot be made undamped by any choice of $\beta,\gamma,\mu$, since $\operatorname{Re}\lambda^{\mathrm{SIR}}_{\pm}=-\mu\Rzero/2<0$ always. The self-sustained cycles of \Cref{sec:hopf} therefore have no counterpart in the control and are attributable to the behavioral mechanism alone. The claim that awareness \emph{generates} epidemic waves is accurate only in that stronger, sustained sense; for the damped transients the accurate statement is that awareness restructures and prolongs waves that demographic renewal already produces.

\subsection{Accessibility is not transmission modulation}
\label{sec:contrast}

We now make precise the relationship between the present framework and the standard formulation in which awareness multiplies the transmission rate.

\begin{proposition}[Fast-relaxation limit]
\label{prop:fastlimit}
Replace $(\kappa_1,\kappa_2,\delta_1,\delta_2)$ by $(\kappa_1,\kappa_2,\delta_1,\delta_2)/\varepsilon$ and let $\varepsilon\to0^{+}$. The accessibility variables are then fast and, at fixed $S$ and $I$, relax to
\begin{equation}
\frac{\Sa}{S}\to g_S(A)=\frac{\kzt(1-A)}{\kzo A+\kzt(1-A)},\qquad
\frac{\Ia}{I}\to g_I(A)=\frac{\dzt(1-A)}{\dzo A+\dzt(1-A)},
\end{equation}
with fast eigenvalues $-(\kappa_1+\kappa_2)/\varepsilon$ and $-(\delta_1+\delta_2)/\varepsilon$. Provided $\kzo,\kzt,\dzo,\dzt>0$, these satisfy $\kappa_1(A)+\kappa_2(A)\ge\min(\kzo,\kzt)>0$ and $\delta_1(A)+\delta_2(A)\ge\min(\dzo,\dzt)>0$ uniformly in $A\in[0,1]$, so both are negative and bounded away from zero; this positivity is assumed throughout the present subsection, and we note that the wider sampling ranges of \Cref{sec:hopf} admit the boundary value zero. Tikhonov's theorem \citep{Verhulst2005} then gives, on finite time intervals, the reduced system
\begin{equation}
\frac{dS}{dt}=\mu-b(A)SI-\mu S,\qquad
\frac{dI}{dt}=b(A)SI-(\gamma+\mu)I,\qquad
\frac{dA}{dt}=\rho I(1-A)-\omega A,
\label{eq:reduced}
\end{equation}
with $b(A)=\beta\,g_S(A)\,g_I(A)$ strictly decreasing on $[0,1]$.
\end{proposition}

System \eqref{eq:reduced} is exactly the classical awareness model in which risk perception scales the transmission rate. The class of transmission-modulating awareness models is therefore the singular limit of the accessibility framework obtained when behavioral relaxation is instantaneous, and every difference between the two must originate in finite behavioral residence times.

One structural consequence of Proposition \ref{prop:fastlimit} sharpens this statement. Dividing numerator and denominator in $g_S$ by $\kzt$ and in $g_I$ by $\dzt$ shows that $g_S$ depends on the susceptible rates only through $\kzo/\kzt$ and $g_I$ on the infectious rates only through $\dzo/\dzt$. The reduced system \eqref{eq:reduced} therefore has no dependence whatever on the absolute scale of the accessibility rates: only their ratios survive, and the absolute scale is precisely what sets the behavioral residence times that the limit destroys. This is directly visible in the numerical results below. The level sets of the leading eigenvalue for the reduced model must be rays through the origin of the $(\dzt,\dzo)$ plane, whereas those of the full model need not be, and \Cref{fig:fig8}a,b confirms exactly this contrast: along a ray of fixed $\dzo/\dzt$, the reduced model's leading eigenvalue is constant to machine precision, whereas the full model's leading eigenvalue varies along the same ray and may change sign. The next theorem shows that the difference is not cosmetic.

\begin{theorem}[Memoryless awareness cannot destabilize the endemic state]
\label{thm:memoryless}
Consider
\begin{equation}
\frac{dS}{dt}=\mu-b(A)SI-\mu S,\quad
\frac{dI}{dt}=\big(b(A)S-(\gamma+\mu)\big)I,\quad
\frac{dA}{dt}=\Psi(I,A),
\end{equation}
where $b\in C^{1}$ is positive and nonincreasing and $\Psi\in C^{1}$ satisfies $\partial_I\Psi>0$ and $\partial_A\Psi<0$. Then every endemic equilibrium with $I^{*}>0$ is locally asymptotically stable.
\end{theorem}

\begin{proof}
Write $q=b(A^{*})I^{*}>0$, $a_1=q+\mu$, $c=\gamma+\mu$, $m=-b'(A^{*})S^{*}I^{*}\ge0$, $p=\partial_I\Psi>0$, $w=-\partial_A\Psi>0$, and use $b(A^{*})S^{*}=c$. The Jacobian in $(S,I,A)$ is
\begin{equation}
J=\begin{pmatrix}-a_1 & -c & m\\ q & 0 & -m\\ 0 & p & -w\end{pmatrix},
\end{equation}
with characteristic polynomial $\lambda^{3}+A_1\lambda^{2}+A_2\lambda+A_3$ where
\begin{equation}
A_1=a_1+w,\qquad A_2=cq+a_1w+mp,\qquad A_3=cwq+m\mu p,
\end{equation}
all strictly positive, and
\begin{equation}
A_1A_2-A_3=a_1cq+a_1^{2}w+a_1w^{2}+mp\,(q+w)>0 .
\end{equation}
The Routh--Hurwitz conditions hold for all admissible values, so all eigenvalues have negative real part. The $R$ equation decouples with eigenvalue $-\mu$.
\end{proof}

\Cref{thm:memoryless} is deliberately stated for a broad class: any prevalence-driven scalar awareness variable acting multiplicatively on transmission produces a stable endemic state, whatever the shape of $b$ and whatever the awareness kinetics, provided awareness rises with prevalence and decays with itself. Sustained oscillations in that class therefore require something extra, and in the literature that extra ingredient is always a lag: distributed delays or higher-order memory kernels in the information index \citep{dOnofrio2009}, delayed information availability \citep{Weitz2020,Mahmud2025}, or two-level awareness with imported cases \citep{Juher2020,Juher2023}. \Cref{sec:hopf} shows that the accessibility formulation supplies that lag endogenously, through behavioral residence times, and can destabilize the endemic equilibrium with no delay imposed anywhere in the model.

\subsection{Loss of stability of the endemic equilibrium}
\label{sec:hopf}

The characteristic polynomial at $E^{*}$ has coefficients that depend on the parameters through the implicitly defined equilibrium, so a closed-form Routh--Hurwitz criterion is not available. We therefore evaluated the spectrum of $J(E^{*})$ numerically after locating $E^{*}$ by bisection on the strictly monotone function $G$, which makes the equilibrium computation exact to machine precision and free of the ambiguity associated with long transient simulation.

For the baseline parameters of \Cref{tab:params} the eigenvalues of the full six-dimensional Jacobian are $-0.420$, $-0.371$, $-0.020$, $-0.019$ and $-0.0086\pm0.0601i$; the value $-0.020=-\mu$ is the eigenvalue of the direction transverse to the invariant hyperplane $N=1$ and plays no role in the dynamics on $\Gam$. The endemic equilibrium is stable, and the leading complex pair predicts damped waves with period $2\pi/0.0601\approx105$ days and relaxation time $1/0.0086\approx116$ days, which matches the transient structure seen in \Cref{fig:fig2} and in \Cref{fig:fig3}a.

Stability is, however, not universal. We drew $2\times10^{5}$ parameter sets uniformly from the ranges in \Cref{tab:params} and the same number from the wider ranges $\beta\in[0.25,3]$, $\gamma\in[0.05,0.5]$, $\mu\in[0.002,0.1]$, $\kzo,\kzt,\dzo,\dzt\in[0,1]$, $\rho\in[0.01,3]$, $\omega\in[10^{-4},1]$. Because subcritical draws admit no endemic equilibrium, all proportions below are taken relative to the draws with $\Rzero>1$, which were $98.9\%$ and $96.3\%$ of the two samples respectively. On that basis $0.034\%$ of the \Cref{tab:params} draws and $0.21\%$ of the wider draws gave an unstable $E^{*}$. The instability is rare, with e-folding times of weeks to months: the largest sample values of $\max\operatorname{Re}\lambda$ were $4.3\times10^{-3}$ over the \Cref{tab:params} ranges, an e-folding time of about $230$ days, and $1.8\times10^{-2}$ over the wider ranges, about $56$ days. These are maxima over the sampled draws rather than suprema over the boxes, and the supremum can exceed them by a factor of order two. Instability is confined to a well-defined corner of parameter space: fast withdrawal, slow return, strong awareness response, and weak fatigue. That corner is not merely a curiosity of the wider ranges, since it intersects the ranges of \Cref{tab:params}. In particular, at $\kzt=0.01$ with all other parameters at baseline, which is the panel shown in \Cref{fig:fig3}b, the leading eigenvalue of $J(E^{*})$ is $+0.0014\pm0.0374i$: the endemic equilibrium is unstable and the barely damped recurrent waves visible over the simulation window in that panel are approaching a periodic attractor rather than decaying to $E^{*}$. The corresponding e-folding time is about $700$ days, comparable to the $800$-day window plotted, which is why that panel appears marginal rather than obviously divergent. Slow susceptible return is the counterpart, on the susceptible side, of the slow infectious return that drives the transition analyzed below.

Fixing $\beta=1.8$, $\dzt=0.01$, $\rho=1$, $\omega=0.01$ and the baseline $\kappa$ values, the leading pair is $-0.00366\pm0.1110i$ at $\dzo=0.30$ and $+0.00251\pm0.1115i$ at $\dzo=0.50$, while $\det J(E^{*})$ remains positive and bounded away from zero along the whole path. No real eigenvalue can therefore cross, and the loss of stability occurs through a complex pair crossing the imaginary axis at $\dzo\approx0.4038$, that is, through a Hopf bifurcation with period $2\pi/0.1114\approx56$ days at onset. Exactly one conjugate pair crosses: at every point of the path examined, at most two eigenvalues have positive real part and they are complex conjugates, the remaining four staying real and bounded away from the imaginary axis.

Direct numerical integration supports supercriticality more sharply than the continuity of the amplitude alone. Near a supercritical Hopf point the attractor amplitude grows as the square root of the distance from onset, so the squared amplitude should be linear in $\dzo$ near $\delta^{0}_{1,\mathrm{c}}=0.4038$. Writing $\mathcal{M}$ for the range of $I$ on the attractor, we obtain $\mathcal{M}^{2}/(\dzo-\delta^{0}_{1,\mathrm{c}})=0.052,\,0.051,\,0.048$ and $0.046$ at $\dzo=0.41,\,0.42,\,0.44$ and $0.46$: nearly constant close to onset, with the drift expected from higher-order terms further away. A local fit gives $\mathcal{M}\approx0.23\,(\dzo-\delta^{0}_{1,\mathrm{c}})^{1/2}$. A subcritical branch turning around near the crossing would not produce this scaling, so the square-root law is substantive evidence rather than a restatement of continuity. No coexisting periodic attractor was detected on the stable side over the initial conditions examined. The numerical results are therefore consistent with a supercritical Hopf bifurcation and provide no evidence of bistability in the explored parameter regime, though we do not claim a normal-form computation.

The mechanism is delayed re-entry generated by finite behavioral relaxation. An accessibility compartment stores individuals who withdraw when awareness is high and releases them according to the awareness-dependent return rates $\kappa_2(A)=\kzt(1-A)$ and $\delta_2(A)=\dzt(1-A)$. Thus, $1/\kzt$ and $1/\dzt$ define baseline relaxation timescales, while at fixed $A<1$ the corresponding instantaneous return timescales are $[\kzt(1-A)]^{-1}$ and $[\dzt(1-A)]^{-1}$. Because re-entry is not instantaneous, the accessibility compartments introduce an endogenous phase lag between prevalence, awareness, and effective contact. As awareness declines, previously withdrawn individuals progressively return to the accessible population, replenishing the pool available for transmission. In the fast-relaxation limit, this storage effect vanishes, the endogenous lag disappears, and \Cref{thm:memoryless} applies. This is the concrete sense in which awareness-driven accessibility is not merely a reparameterization of awareness-driven transmission reduction.

The lag just described is not the only lever on stability, and separating it from the feedback gain clarifies what destabilizes the endemic equilibrium. Driving awareness by $\Ia$ rather than by $I=\Ia+\Iu$ places the awareness signal downstream of the accessibility lag, which might be expected to promote oscillation further; instead it reduces the feedback gain by the factor $1/\Phi(A^{*})$, and that effect dominates. Under $\rho\Ia(1-A)$ the leading eigenvalue is $-0.0110\pm0.0607i$ at baseline, against $-0.0086\pm0.0601i$ under $\rho I(1-A)$, and at $\dzo=0.60$ on the path above it is $-0.0148\pm0.1216i$, so the equilibrium remains stable exactly where the $\rho I$ convention gives a limit cycle. The mechanism is nonetheless not an artifact of the awareness-source convention: sampling the wider ranges above under $\rho\Ia(1-A)$ still yields instability through a complex pair in the same corner of parameter space, in $0.047\%$ of draws with $\Rzero>1$ and $\max\operatorname{Re}\lambda=9.9\times10^{-3}$. One sub-case is convention-free by construction: when $\dzo=0$ the class $\Iu$ is empty at equilibrium and decouples in the linearization, so the two awareness sources give identical equilibria and identical spectra, and $E^{*}$ can still lose stability; at baseline with $\dzo=0$ and $\kzt=0.01$ the leading pair is $+0.0011\pm0.0366i$ under either convention, so storage in the susceptible accessibility class alone suffices for the Hopf instability. What destabilizes the endemic equilibrium is therefore not the presence of a lag as such but its combination with sufficient feedback gain; the two act as separate, and here opposing, levers, and the accessible-prevalence convention weakens the gain enough to make destabilization harder without removing it.

\begin{remark}
Because $E^{*}$ is unique for all $\Rzero>1$ (\Cref{thm:endemic}) and
infection is uniformly persistent (\Cref{thm:persistence}), loss of stability
of $E^{*}$ cannot lead to convergence to the disease-free equilibrium or to
a second endemic equilibrium. The numerically observed stable periodic orbit
is therefore consistent with the persistence and uniqueness results, although
a complete analytical characterization of the post-bifurcation attractor is
beyond the scope of the present study.
\end{remark}

\section{Numerical results}
\label{sec:results}

\subsection{Numerical methods}

All numerical computations were performed in Python~3.12.3 using
NumPy~2.4.4, SciPy~1.17.1, and Matplotlib~3.10.8. Trajectories were
integrated with \texttt{scipy.integrate.solve\_ivp} using the LSODA
method, with relative and absolute tolerances of $10^{-10}$ and
$10^{-12}$, respectively. Unless stated otherwise, simulations were
initialized at
\[
\Sa(0)=0.9,\qquad
\Su(0)=0.05,\qquad
\Ia(0)=0.05,\qquad
\Iu(0)=R(0)=A(0)=0,
\]
and all parameters were fixed at the baseline values in
\Cref{tab:params}.

Endemic equilibria were computed independently of the time integration
by bisection of the strictly decreasing scalar function $G$ in
\eqref{eq:G}. Existence and uniqueness of the root are guaranteed by
\Cref{thm:endemic}. Each computed equilibrium was verified by direct
substitution into \eqref{eq:1}--\eqref{eq:6}, with residuals at machine
precision. Local stability was determined from the eigenvalues of the
analytical Jacobian evaluated at the computed equilibrium.

The equilibrium maps in \Cref{fig:fig6} were evaluated on
$140\times140$ uniform parameter grids, and the stability maps in
\Cref{fig:fig8}a--c on $90\times90$ grids, with the equilibrium
recomputed at every grid point. For \Cref{fig:fig8}d, trajectories were
integrated to $t=3\times10^{4}$ and attractor extrema were recorded only
after the transient had been discarded. Possible multistability was
examined from 25 distinct initial conditions spanning the feasible
region $\Gam$; no coexisting attractor was detected in the regimes
examined. This numerical observation does not exclude coexistence of
an equilibrium and a periodic orbit elsewhere in parameter space.

No dedicated continuation package was used. Stability boundaries were
located by dense evaluation of the leading eigenvalue along the unique
equilibrium branch. Accordingly, the results below establish a
numerically identified stability crossing and provide evidence
consistent with a supercritical Hopf transition; no normal-form
calculation or numerical continuation of the periodic branch is
claimed.

\subsection{No-awareness control and transient epidemic waves}

The transient waves observed at baseline are not created by awareness
alone. As shown analytically in \Cref{sec:sircontrol}, setting $\rho=0$
eliminates awareness generation and the system asymptotically reduces
to the corresponding SIR model with demography, whose period, relaxation
time, and damping ratio are compared with the full model there. The
direct comparison in \Cref{fig:control} confirms both numerically:
awareness lengthens the oscillatory timescale and substantially
reduces damping, causing recurrent waves to remain visible for longer,
but it is not the source of the underlying transient oscillation.

At baseline, the full model converges to
\[
I^{*}=0.0155,
\qquad
I_{\mathrm{SIR}}^{*}=0.0576,
\]
corresponding to a $73\%$ reduction in endemic prevalence without any
change in transmissibility per contact. The mechanism is behavioral:
increasing prevalence raises awareness and promotes withdrawal from
accessible states; as prevalence falls, fatigue erodes awareness and
individuals return to accessibility. This storage-and-release process
modifies both the transient waves and the endemic burden. Trajectory
sweeps in the awareness response and in the susceptible accessibility
rates, together with the global sensitivity analysis and the transient
diagnostics, are provided in \Cref{app:numerical}.

\begin{figure}[!htbp]
\centering
\includegraphics[
    width=0.88\textwidth,
    keepaspectratio
]{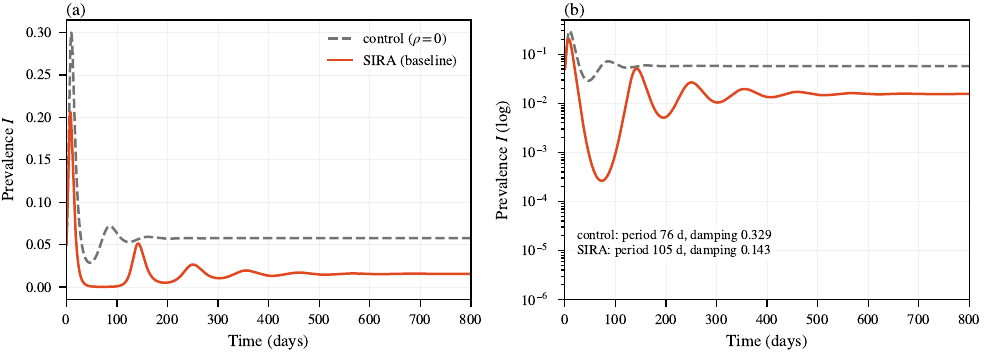}
\caption{\textbf{No-awareness control.}
Prevalence trajectories for the full awareness model and the
$\rho=0$ control under identical initial conditions and baseline
parameters. (a) Prevalence on a linear scale. (b) The same trajectories
on a logarithmic scale, highlighting the decay of successive waves.
The $\rho=0$ control has an oscillation period of approximately
$76$ days and damping ratio $0.329$, compared with approximately
$105$ days and $0.143$, respectively, for the full model. Thus,
awareness lengthens the oscillatory period and reduces damping but does
not generate the underlying transient oscillation.}
\label{fig:control}
\end{figure}

\subsection{Endemic effects of awareness, fatigue, and transmission}

The equilibrium computations confirm the analytical comparative
statics. Endemic prevalence decreases monotonically with the awareness
response $\rho$ and increases with behavioral fatigue $\omega$, while
endemic awareness changes in the opposite direction
(\Cref{fig:fig5}a,b,d,e). At fixed $\omega$, increasing $\rho$ decreases
the ratio
\[
\theta=\frac{\omega}{\rho},
\]
so that $A^{*}$ approaches the limiting value
$A_{\infty}=0.514$ and $I^{*}$ approaches zero according to the
low-fatigue asymptotics in \Cref{thm:lowfatigue}. Increasing fatigue
reverses this effect by accelerating return to accessible states.

Increasing the transmission rate $\beta$ raises both endemic prevalence
and endemic awareness (\Cref{fig:fig5}c,f), reflecting the feedback
between infection and perceived risk. Both endemic branches emerge
continuously from zero at
\[
\beta=\gamma+\mu,
\qquad \Rzero=1,
\]
consistent with the analytical result that behavioral parameters
modify post-invasion dynamics but do not shift the invasion threshold.

\begin{figure}[!ht]
\centering
\includegraphics[
    width=0.96\textwidth,
    height=0.70\textheight,
    keepaspectratio
]{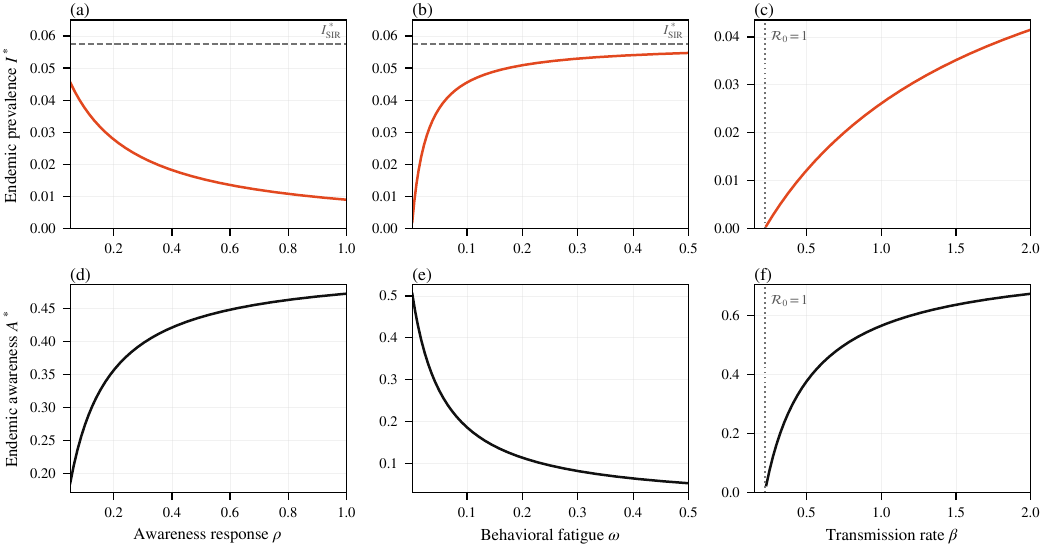}
\caption{\textbf{Endemic prevalence and awareness as functions of
behavioral and transmission parameters.}
(a--c) Endemic prevalence $I^{*}$ as a function of awareness response
$\rho$, behavioral fatigue $\omega$, and transmission rate $\beta$.
(d--f) Corresponding endemic awareness $A^{*}$. Remaining parameters
take the baseline values in \Cref{tab:params}. The dashed horizontal line
in panels (a) and (b) marks the SIR benchmark
$I^{*}_{\mathrm{SIR}}=0.0576$; $I^{*}$ approaches it as $\rho\to0$ and as
$\omega$ grows, which is the expected consistency check against the
$\rho=0$ control. Panels (c) and (f)
identify $\beta=\gamma+\mu=0.22$, corresponding to $\Rzero=1$.
The endemic branches emerge continuously from zero at this threshold.}
\label{fig:fig5}
\end{figure}

\subsection{Joint parameter dependence and the fatigue-to-response ratio}

The two-parameter equilibrium maps in \Cref{fig:fig6} reveal the
structure predicted by \Cref{thm:theta}. In the $(\omega,\rho)$ plane,
the level sets of both $I^{*}$ and $A^{*}$ follow rays of constant
$\omega/\rho$, confirming that the endemic equilibrium depends on
awareness response and fatigue only through
\[
\theta=\frac{\omega}{\rho}.
\]
Thus, proportional changes in $\rho$ and $\omega$ leave the endemic
state unchanged even though, as shown below, they need not leave its
stability unchanged.

This invariance does not extend to parameter planes involving
$\beta$. Stronger awareness response is most effective at suppressing
endemic prevalence when transmission is low to moderate
(\Cref{fig:fig6}b), whereas increasing fatigue moves prevalence toward
the corresponding unmitigated level (\Cref{fig:fig6}c). High endemic
awareness occurs when response is strong and fatigue is weak
(\Cref{fig:fig6}d--f).

\begin{figure}[!ht]
\centering
\includegraphics[
    width=\textwidth,
    height=0.60\textheight,
    keepaspectratio
]{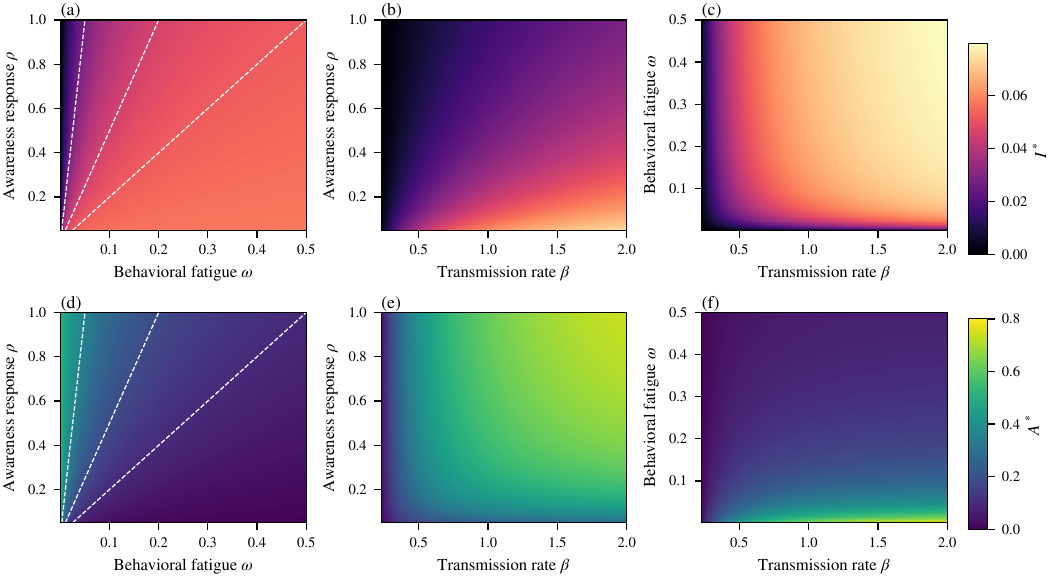}
\caption{\textbf{Joint dependence of the endemic equilibrium on
behavioral and transmission parameters.}
(a--c) Endemic prevalence $I^{*}$ as a function of
$(\omega,\rho)$, $(\beta,\rho)$, and $(\beta,\omega)$, respectively.
(d--f) Corresponding endemic awareness $A^{*}$. White dashed lines in
panels (a) and (d) denote selected rays of constant $\rho/\omega$.
The alignment of the level sets with these rays confirms the
$\theta=\omega/\rho$ reduction established in \Cref{thm:theta}.
Remaining parameters take the baseline values in \Cref{tab:params}.}
\label{fig:fig6}
\end{figure}

\subsection{Verification of the endemic scaling and burden decomposition}

\Cref{fig:fig7} provides direct numerical verification of the main
analytical equilibrium results. Endemic prevalence computed
independently for $\rho\in\{0.1,0.5,2\}$ collapses onto a single curve
when expressed as a function of $\theta=\omega/\rho$
(\Cref{fig:fig7}a). In the low-fatigue regime, the numerical solution
also follows the asymptotic relation
\[
I^{*}\sim
\theta\frac{A_{\infty}}{1-A_{\infty}},
\qquad \theta\to0,
\]
over several orders of magnitude (\Cref{fig:fig7}b).

The exact decomposition in \eqref{eq:decomp} partitions the reduction
relative to the SIR benchmark into susceptible withdrawal and
infectious sequestration (\Cref{fig:fig7}c). At baseline, susceptible
withdrawal accounts for approximately $79\%$ of the reduction and
infectious sequestration for approximately $21\%$. Both contributions
weaken as fatigue increases. Finally, \Cref{fig:fig7}d shows that the
relative burden $I^{*}/I_{\mathrm{SIR}}^{*}$ increases with
transmission intensity for fixed $\theta$, indicating that behavioral
protection becomes less effective in relative terms as transmission
strengthens.

\begin{figure}[!ht]
\centering
\includegraphics[
    width=0.86\textwidth,
    height=0.70\textheight,
    keepaspectratio
]{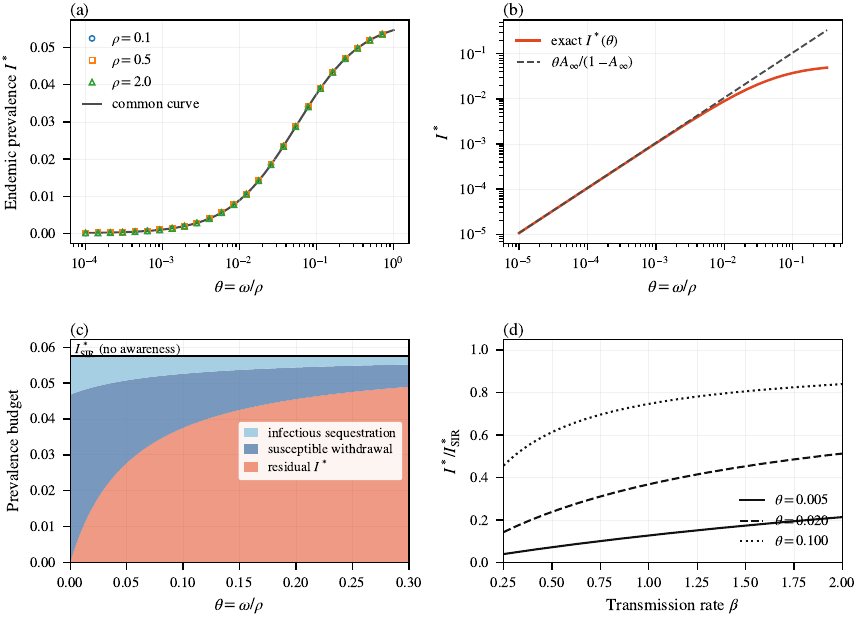}
\caption{\textbf{Numerical verification of the analytical endemic
results.}
(a) Endemic prevalence for $\rho=0.1,0.5,2$ collapses onto a common
curve in $\theta=\omega/\rho$. (b) Exact endemic prevalence compared
with the low-fatigue asymptotic law
$I^{*}\sim\theta A_{\infty}/(1-A_{\infty})$.
(c) Exact decomposition of the SIR endemic burden into residual
prevalence, susceptible withdrawal, and infectious sequestration.
(d) Relative endemic burden $I^{*}/I_{\mathrm{SIR}}^{*}$ as a function
of $\beta$ for three values of $\theta$. Remaining parameters take the
baseline values in \Cref{tab:params}.}
\label{fig:fig7}
\end{figure}

\subsection{From damped waves to sustained cycles}

The baseline parameter set lies in the stable regime, so its recurrent
waves are transient and trajectories approach $E^{*}$. This behavior
is not universal. The stability map in \Cref{fig:fig8}a shows that the
endemic equilibrium loses stability when infectious withdrawal is
sufficiently rapid relative to return. In contrast, the corresponding
fast-relaxation model remains locally stable throughout the same
parameter region (\Cref{fig:fig8}b), consistent with
\Cref{thm:memoryless}. Thus, finite residence in the behavioral states
introduces a dynamical mechanism that is absent from the instantaneous
transmission-modulating limit.

Along the representative path used in \Cref{fig:fig8}d, the leading
complex-conjugate eigenvalues cross the imaginary axis near
\[
\dzo\simeq0.404.
\]
The numerically observed attractor amplitude emerges continuously from
the equilibrium at this boundary and follows the square-root law
characteristic of a supercritical transition, with
$\mathcal{M}^{2}/(\dzo-\delta^{0}_{1,\mathrm{c}})$ nearly constant close to onset
(\Cref{sec:hopf}). Below the boundary, trajectories
display damped oscillations and converge to $E^{*}$
(\Cref{fig:fig8}e); above it, they approach a persistent periodic
orbit (\Cref{fig:fig8}f). No coexisting attractor was detected among
the initial conditions examined.

The $(\omega,\rho)$ stability map in \Cref{fig:fig8}c shows that
instability is favored by strong awareness response combined with weak
fatigue. Importantly, the equilibrium reduction to
$\theta=\omega/\rho$ does not extend to stability: $\rho$ and $\omega$
enter the Jacobian separately, so parameter combinations producing the
same endemic equilibrium can have different local dynamics.

\begin{figure}[!ht]
\centering
\includegraphics[
    width=\textwidth,
    height=0.70\textheight,
    keepaspectratio
]{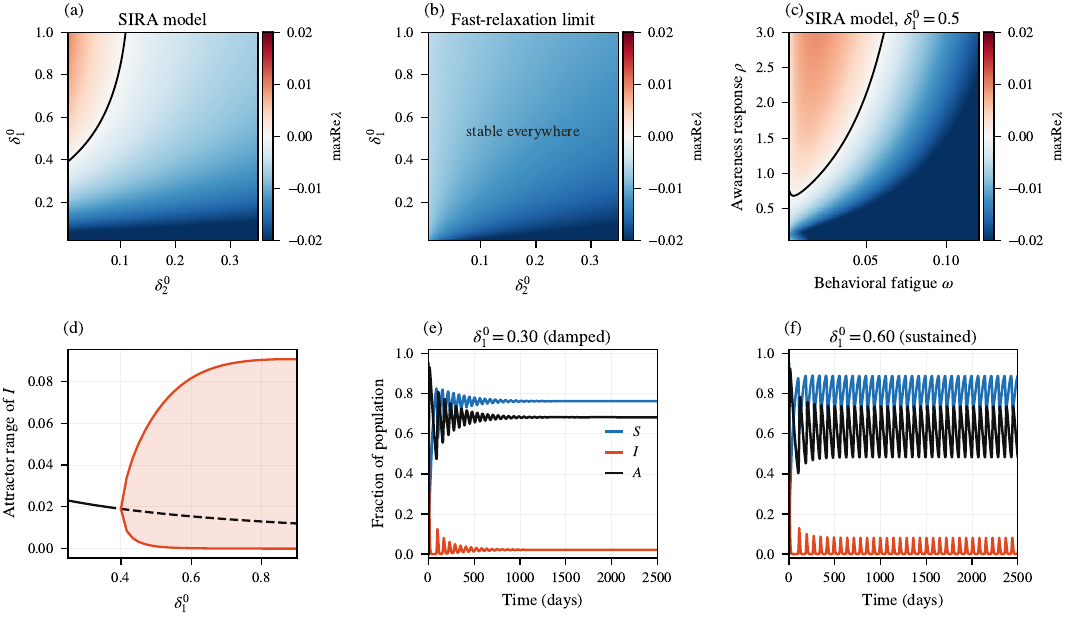}
\caption{\textbf{Finite behavioral residence times can destabilize the
endemic equilibrium.}
(a) Leading eigenvalue $\max\operatorname{Re}\lambda$ of $J(E^{*})$
over the $(\dzt,\dzo)$ plane at
$\beta=1.8$, $\rho=1$, $\omega=0.01$, and
$\kzo=\kzt=0.4$; the black curve marks the stability boundary.
(b) Corresponding calculation for the fast-relaxation limit
\eqref{eq:reduced}; the endemic equilibrium remains locally stable
throughout the same parameter domain, and its level sets are rays through
the origin because that limit depends on $\dzo$ and $\dzt$ only through
their ratio (\Cref{sec:contrast}).
(c) Leading eigenvalue over the $(\omega,\rho)$ plane at $\dzo=0.5$, with
$\beta=1.8$, $\kzo=\kzt=0.4$ and $\dzt=0.01$ fixed.
(d) Equilibrium branch and numerically observed attractor extrema as
$\dzo$ varies; the equilibrium is shown solid where stable and dashed
where unstable.
(e,f) Representative trajectories below and above the stability
boundary at $\dzo=0.30$ and $\dzo=0.60$, respectively.
The emerging periodic branch is consistent with a supercritical Hopf
transition; no normal-form calculation or numerical continuation of
the periodic branch is claimed.}
\label{fig:fig8}
\end{figure}

In the cycling regime, prevalence falls to very low levels during
inter-wave troughs: on the computed attractors the minimum of $I$ is
$2.4\times10^{-4}$ at $\dzo=0.60$ (\Cref{fig:fig8}f), $2.9\times10^{-5}$ at
$\dzo=0.90$, and $3.3\times10^{-6}$ on the $\kzt=0.01$ cycle of
\Cref{fig:fig3}b, whose period of about $241$ days exceeds the $168$ days
of the linearization and whose peaks within the plotted window are still
decaying toward the cycle maximum of $0.022$. In a finite population, such low prevalence would
increase the possibility of stochastic fade-out. The deterministic
cycles should therefore be interpreted as a mechanism for recurrent
behaviorally generated epidemic waves rather than as a prediction that
infection necessarily persists through every trough in a finite
population.


\section{Discussion}
\label{sec:discussion}

\subsection*{Behavior as a dynamical state}

Behavioral epidemic models commonly represent risk response by reducing the transmission coefficient \citep{Funk2009,dOnofrio2009,Fenichel2011,Weitz2020,Juher2023,Mahmud2025}. Our analysis clarifies what is lost under this representation. Proposition \ref{prop:fastlimit} shows that transmission modulation arises as the instantaneous-relaxation limit of the accessibility model, when behavioral withdrawal has no residence time. Within the scalar, prevalence-driven feedback class considered here, \Cref{thm:memoryless} shows that this limit cannot destabilize the endemic equilibrium. Recurrent waves in related models therefore require additional temporal structure, such as memory kernels \citep{dOnofrio2009}, delayed information \citep{Weitz2020,Mahmud2025}, or multiple awareness levels \citep{Juher2020,Juher2023}. The comparison with \citet{Juher2020,Juher2023} is worth drawing explicitly, since those models produce behaviorally induced limit cycles and even saddle-node bifurcations of limit cycles. Their two-level awareness structure carries an effective delay: an individual must pass through an intermediate awareness state before reaching the fully protective one, and that intermediate residence time plays the role of a lag. Our formulation has no such structure in the awareness variable, which responds to prevalence instantaneously through \eqref{eq:6}; the lag resides instead in the accessibility compartments. The two mechanisms are therefore distinct realizations of the same requirement, and the present model shows that the lag can be carried by behavioral occupancy alone, without any staging of awareness itself. In contrast, the accessibility model generates an endogenous lag through finite residence times in behavioral states. Numerical bifurcation analysis shows that this mechanism alone can produce a supercritical Hopf bifurcation and sustained epidemic oscillations.

This distinction provides a testable rationale for state-based representations \citep{Epstein2008,Perra2011} when epidemic recurrence is of interest. It also identifies a measurable quantity often absent from transmission-modulating models: the duration of protective behavior, rather than only the strength of the initial response.

\subsection*{Coexistence and self-limiting awareness}

Uniform persistence (\Cref{thm:persistence}) establishes sustained coexistence of infection and awareness when $\Rzero>1$. Because awareness is generated by prevalence, however, the feedback is intrinsically self-limiting: declining prevalence weakens the behavioral response responsible for that decline. Awareness can therefore suppress endemic infection but cannot eliminate it within this model.

The decomposition \eqref{eq:decomp} quantifies this suppression. At baseline, behavior reduces endemic prevalence by $73\%$ relative to the corresponding SIR model, with $79\%$ of the reduction attributable to susceptible withdrawal and $21\%$ to infectious sequestration.

The decomposition also clarifies the sensitivity structure of the model. Because $\Rzero=\beta/(\gamma+\mu)$ depends only on $\beta$, $\gamma$, and $\mu$, the more informative sensitivity analysis concerns the endemic quantities $I^{*}$ and $A^{*}$, which also depend on the behavioral parameters. Analytically, \Cref{thm:monotone} establishes the direction of each behavioral effect on $I^{*}$ over the admissible parameter space, while \Cref{thm:theta} shows that the equilibrium dependence on awareness response and fatigue reduces to the ratio $\theta=\omega/\rho$. To complement these analytical results, we performed a Latin-hypercube/partial-rank-correlation analysis using $N=1500$ parameter sets sampled over the ranges in \Cref{tab:params}. The resulting PRCCs for $I^{*}$ and $A^{*}$ are reported in \Cref{fig:prcc}, with the sampling conditions stated in \Cref{app:numerical}; their signs agree with the analytical comparative statics, while their magnitudes quantify the relative influence of the parameters over the sampled ranges. Thus, the analytical and global sensitivity analyses are complementary: the former establishes parameter-effect directions exactly, whereas the latter quantifies their relative importance over the parameter domain considered.

\subsection*{Fatigue as a long-term constraint}

\Cref{thm:theta,thm:lowfatigue} show that the endemic state depends on awareness response and fatigue only through $\theta=\omega/\rho$. Moreover, $I^{*}$ decreases linearly with $\theta$ as $\theta\to0$, while awareness approaches a limiting plateau. Thus, faster awareness response and slower behavioral fatigue have equivalent effects on the equilibrium through their ratio, although they need not have equivalent transient or stability effects.

The low-fatigue asymptotics further show that small increases in fatigue can produce substantial increases in endemic prevalence near $\theta=0$. The durability of protective behavior is therefore a key determinant of long-run disease burden, consistent with empirical evidence emphasizing the persistence of behavioral responses during prolonged epidemics \citep{Jorgensen2022,Chen2024}.

\subsection*{Implications for epidemic control}

Threshold neutrality  (Proposition \ref{prop:R0}) limits the role of awareness in this model: because awareness is generated only by prevalence, it vanishes at the disease-free state and cannot alter the invasion threshold. This reflects the endogenous construction rather than a general property of behavioral models; with an exogenous awareness source as in \eqref{eq:exogenous}, behavior enters the threshold through \eqref{eq:R0eta}. The effect is not marginal: at the baseline parameters a standing information source of $\eta=0.02$~day$^{-1}$ holds disease-free awareness at $A^{\dagger}=2/3$, above both $A^{*}=0.437$ and $A_{\infty}=0.514$, and gives $\Rzero^{\eta}=0.62$, so awareness maintained independently of current prevalence at that level moves this parameter set below threshold outright. The threshold itself has a closed form. Setting $\Rzero^{\eta}=1$ in \eqref{eq:R0eta} gives exactly \eqref{eq:Ainf} evaluated at $A=A^{\dagger}$, and the difference of the two sides of \eqref{eq:Ainf} is strictly decreasing in $A$ (proof of \Cref{thm:lowfatigue}), so $\Rzero^{\eta}<1$ if and only if $A^{\dagger}>A_{\infty}$, that is, if and only if $\eta>\omega A_{\infty}/(1-A_{\infty})$; at baseline this gives $\eta_{\mathrm{c}}=0.0106$~day$^{-1}$. A standing information source eliminates infection exactly when it holds awareness above the plateau that prevalence-driven awareness could reach with no fatigue at all. Its principal effect is instead to reduce accessible susceptibility and infectious accessibility after invasion, thereby lowering endemic burden and reshaping epidemic timing. \Cref{fig:fig2}a--d shows that stronger awareness responses deepen and lengthen the low-prevalence intervals between waves. We report this comparison qualitatively because the underlying quantity, the time to resurgence, is not a property of the equilibrium and depends on the initial condition; the equilibrium-level counterpart that we can state exactly is the damping ratio at $E^{*}$, which falls from $0.329$ in the no-awareness control to $0.143$ at baseline (\Cref{sec:sircontrol}).

In practical terms the accessible lever is $\rho$, the speed with which perceived risk tracks prevalence, which risk communication and community engagement are designed to raise \citep{Tambo2021,Emovwodo2022,Maganga2024}, together with the durability of the response, which enters as $\omega$. Awareness should nonetheless complement, rather than replace, interventions that directly reduce the invasion threshold, such as vaccination or treatment. By suppressing transmission after invasion and delaying subsequent waves, behavioral responses may provide additional time for such interventions to be deployed.

Finite behavioral residence times also introduce a potential trade-off. In the numerically identified unstable regime, lower endemic burden can coexist with large recurrent peaks and deep troughs. Equilibrium-based assessments may therefore understate temporal variability. The same combination of rapid withdrawal, slow return, strong awareness response, and weak fatigue that strongly suppresses prevalence can also promote oscillatory dynamics, emphasizing the importance of evaluating both burden and stability.

\subsection*{Limitations and future directions}

The model assumes homogeneous mixing among accessible individuals and a single population-level awareness variable, excluding spatial structure, contact-network heterogeneity, and variation in individual risk responses \citep{Granell2013,Wang2015}. Awareness responds directly to prevalence, omitting reporting delays and media-driven information dynamics. Recovered individuals have no accessibility states, an assumption that would require modification under waning immunity. The deterministic formulation also excludes stochastic fade-out, particularly relevant during the deep troughs predicted in the oscillatory regime.

The demographic rescaling of Remark \ref{rem:rescale} carries a price that should be stated plainly. Since $I^{*}=\theta A^{*}/(1-A^{*})$, an endemic awareness level of order $0.5$ requires $\theta$ of order $I_{\max}$. At $\mu=3.9\times10^{-5}$~day$^{-1}$ this forces $\omega\approx2.1\times10^{-5}$~day$^{-1}$ at $\rho=0.5$, a fatigue timescale of about $130$ years, which is not compatible with the weeks-to-months erosion of protective behavior documented in the literature cited above \citep{Harvey2020,Jorgensen2022,Chen2024}. Under human demography together with empirically plausible fatigue, therefore, $A^{*}\to0$ and the effect of awareness on the \emph{endemic} state largely disappears; what survives is the transient and oscillatory behavior, which is where the present results are in any case most informative. The alternative reading, which we regard as the more useful one, is to treat $\mu$ as an effective susceptible-replenishment rate dominated by waning immunity rather than by births and deaths, in which case a faster replenishment rate is defensible on its own terms; this is consistent with the treatment of recovered accessibility discussed above and would require the model to be extended by an explicit waning channel.

Finally, the Hopf transition and its apparent supercritical character are established numerically; a normal-form analysis of the six-dimensional system remains open.

Natural extensions include heterogeneous accessibility across demographic or socioeconomic groups \citep{Bish2010}, explicit information dynamics with reporting delays, and stochastic formulations that can resolve extinction during low-prevalence periods. Incorporating vaccination or treatment \citep{Bauch2004,Manfredi2009,Betsch2020} would further allow threshold-neutral behavioral responses to be studied jointly with interventions that directly modify the invasion threshold.

\section{Conclusion}
\label{sec:conclusion}

The results of \Cref{sec:analysis,sec:results} support one central claim: representing protective behavior as occupancy of a state with finite residence time, rather than as an instantaneous rescaling of transmission, is not a cosmetic modeling choice. It leaves the invasion threshold untouched, it reduces endemic burden through a mechanism decomposed exactly in \Cref{thm:decomp}, and it is the only source, within the class of models compared in \Cref{sec:contrast}, capable of destabilizing the endemic equilibrium into self-sustained cycles. Behavioral fatigue is what ultimately limits how much of that protection is realized. The duration and persistence of protective behavior, not only its initial strength, are therefore key determinants of both endemic burden and epidemic dynamics.

\section*{Conflict of Interest}
The authors declare no conflict of interest.

\section*{Funding}
This research received no external funding.

\section*{Data and Code Availability}

No data were generated or analyzed in this theoretical study. The source
code used for the numerical analyses, stability and bifurcation computations,
parameter scans, and figures, together with a notebook that reproduces the
numerical results reported in this study, is publicly available on GitHub at
\url{https://github.com/aalsammani/sira-epidemic-awareness} and archived on
Zenodo \citep{Alsammani2026SIRA}, version v1.0.0,
\url{https://doi.org/10.5281/zenodo.22164948}.


\appendix
\section{Supplementary numerical results}
\label{app:numerical}

This section provides additional simulations supporting the principal
results without duplicating the main-text analysis. It includes
one-parameter transient sweeps, direct verification of the accessibility
comparative statics, a global sensitivity analysis, and quantitative
measures of transient wave behavior. Unless stated otherwise, all
parameters take the baseline values in \Cref{tab:params} and all
trajectories use the initial condition specified in
\Cref{sec:results}.

\subsection*{Transient effects of awareness response and fatigue}

\Cref{fig:fig2} illustrates how awareness response reshapes the
transient dynamics: increasing $\rho$ delays and suppresses subsequent
prevalence peaks. Behavioral fatigue $\omega$ acts in the mirror-image
direction, accelerating the return to accessibility and shortening the
transient; \Cref{fig:wave_metrics} quantifies both effects on the same
trajectories. These trajectory-level effects complement the equilibrium
trends established analytically in \Cref{thm:monotone}.

\begin{figure}[!htbp]
\centering
\includegraphics[
    width=.95\textwidth,
    keepaspectratio
]{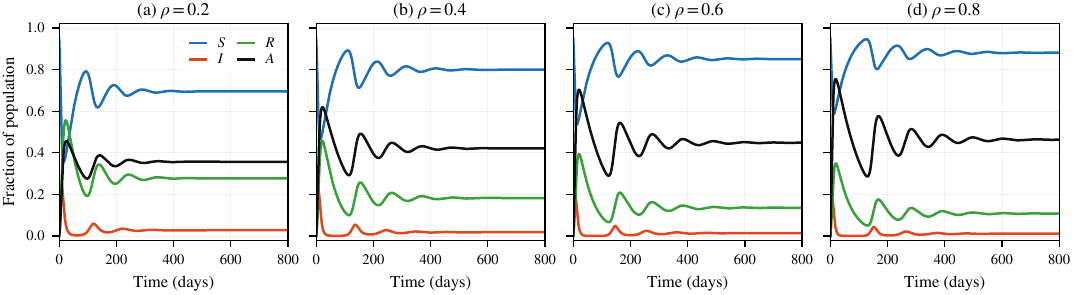}
\caption{\textbf{Awareness response shapes transient epidemic waves.}
Time series of $S$, $I$, $R$, and $A$ over $800$ days for awareness
response $\rho=0.2,0.4,0.6,0.8$ (a--d). Increasing $\rho$ prolongs the
transient wave structure and lowers endemic prevalence; the mirror-image
effect of behavioral fatigue $\omega$ on the same trajectories is
quantified in \Cref{fig:wave_metrics}.}
\label{fig:fig2}
\end{figure}

\subsection*{Accessibility-rate effects}

Susceptible accessibility has a pronounced effect on both transient
damping and endemic burden. Faster withdrawal lowers prevalence,
whereas faster return increases it. The slow-return case
$\kzt=0.01$ lies in the unstable regime identified in
\Cref{sec:hopf}, explaining the persistent recurrent waves in
\Cref{fig:fig3}b.

\begin{figure}[!htbp]
\centering
\includegraphics[
    width=.62\textwidth,
    keepaspectratio
]{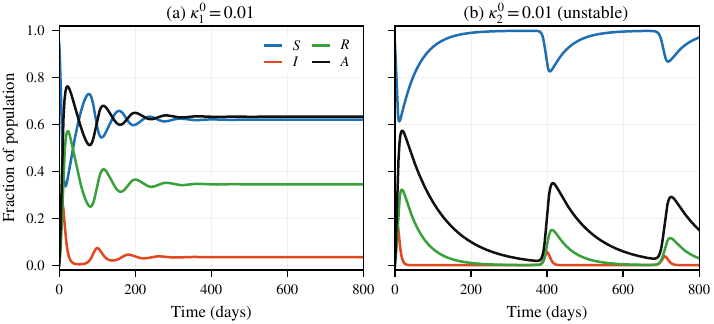}
\caption{\textbf{Susceptible accessibility rates regulate transient
damping.}
Time series of $S$, $I$, $R$, and $A$ over $800$ days at the slowest
swept rate on each side of the withdrawal/return pair: (a)
$\kzo=0.01$, damped; (b) $\kzt=0.01$, which lies in the unstable regime
of \Cref{sec:hopf} and exhibits persistent rather than damped recurrent
waves. The full equilibrium sweep over both rates is given in
\Cref{fig:fig4}a,b.}
\label{fig:fig3}
\end{figure}

The equilibrium sweeps in \Cref{fig:fig4} provide a direct numerical
check of the comparative statics in \Cref{thm:monotone}: increasing
either withdrawal rate decreases $I^{*}$, whereas increasing either
return rate increases $I^{*}$. The stronger response to susceptible
accessibility is consistent with the burden decomposition in
\eqref{eq:decomp}.

\begin{figure}[!ht]
\centering
\includegraphics[
    width=.82\textwidth,
    keepaspectratio
]{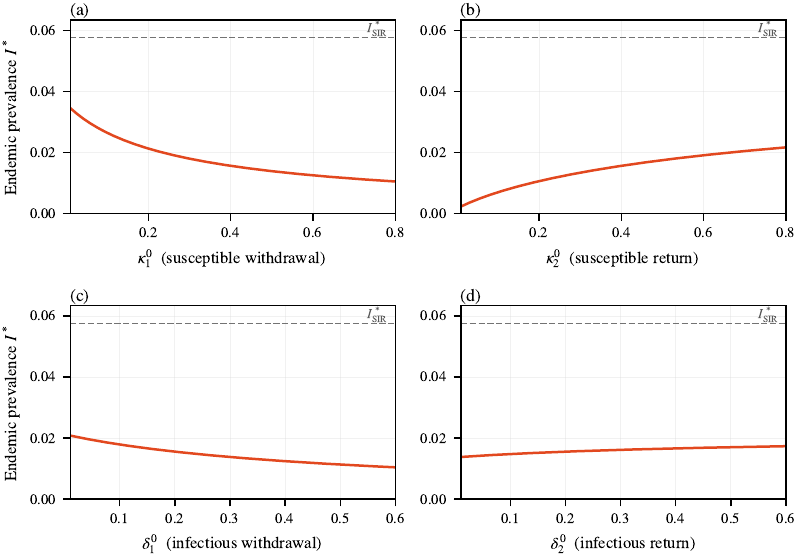}
\caption{\textbf{Endemic prevalence as a function of the accessibility
rates.}
Endemic prevalence $I^{*}$ against
(a) susceptible withdrawal $\kzo$,
(b) susceptible return $\kzt$,
(c) infectious withdrawal $\dzo$, and
(d) infectious return $\dzt$.
The dashed horizontal line denotes the SIR benchmark
$I_{\mathrm{SIR}}^{*}=0.0576$. The numerical directions agree with the
comparative statics established in \Cref{thm:monotone}.}
\label{fig:fig4}
\end{figure}

\subsection*{Global sensitivity of the endemic equilibrium}

To complement the exact comparative statics, we performed a
Latin-hypercube/partial-rank-correlation analysis using $N=1500$
parameter sets sampled over the ranges in \Cref{tab:params}. The seven
swept parameters are $\beta$, $\kzo$, $\kzt$, $\dzo$, $\dzt$, $\rho$ and
$\omega$; $\gamma$ and $\mu$ are held at their baseline values, since
\Cref{tab:params} assigns them no range. Subcritical draws, for which no
endemic equilibrium exists, were assigned $I^{*}=A^{*}=0$; they numbered
$16$ of $1500$, and restricting the analysis to draws with $\Rzero>1$
changes every reported coefficient by less than $0.013$. We note also that
PRCC presupposes a monotone relationship between each input and the output,
and that \Cref{thm:monotone} establishes exactly that monotonicity for
$I^{*}$ over the whole admissible parameter space, so the analytical result
licenses the numerical method rather than merely agreeing with it.
\Cref{fig:prcc} shows that all PRCC signs for $I^{*}$ agree with the analytical
derivatives in \Cref{thm:monotone}. For $A^{*}$ the transmission and
accessibility parameters act only through $I^{*}$ and share its sign, while
$\rho$ and $\omega$ act directly and follow \Cref{thm:theta}, under which
$A^{*}$ decreases in $\theta=\omega/\rho$. For endemic prevalence, $\beta$
has the strongest positive association, while $\rho$ is strongly
negative and $\omega$ strongly positive. The corresponding awareness
indices show the expected reversal of the behavioral effects. Thus,
the global numerical analysis supports the analytical sign results
while additionally providing a ranking of association strengths over
the sampled parameter domain.

\begin{figure}[!htbp]
\centering
\includegraphics[
    width=0.86\textwidth,
    keepaspectratio
]{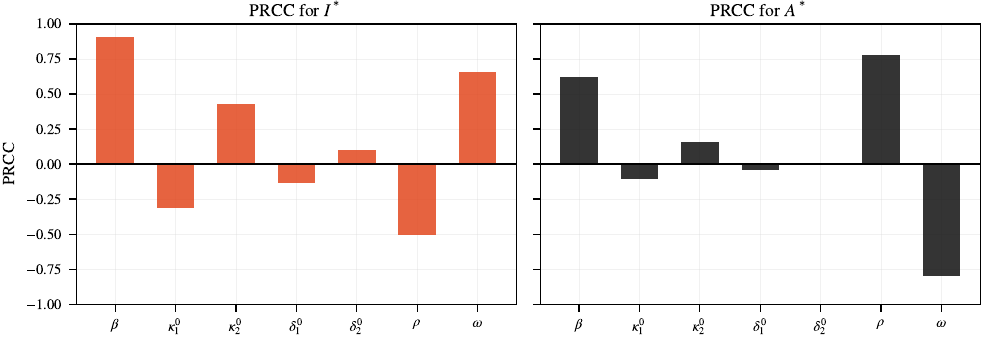}
\caption{\textbf{Global sensitivity of the endemic equilibrium.}
Partial rank correlation coefficients (PRCCs) for endemic prevalence
$I^{*}$ (left) and endemic awareness $A^{*}$ (right), obtained from
$N=1500$ Latin hypercube samples over the seven parameters assigned ranges
in \Cref{tab:params}, with $\gamma$ and $\mu$ held at their baseline values.
Subcritical draws are assigned $I^{*}=A^{*}=0$. The signs are consistent with
the analytical comparative statics.}
\label{fig:prcc}
\end{figure}

\subsection*{Quantitative transient-wave diagnostics}

The qualitative trajectory patterns are summarized quantitatively in
\Cref{fig:wave_metrics}. Increasing awareness response $\rho$ delays
the second prevalence peak, reduces its magnitude and cumulative
incidence, and decreases the damping ratio at the endemic equilibrium.
Increasing fatigue $\omega$ generally reverses these effects by
accelerating behavioral relaxation. The second-peak magnitude is
nonmonotone over part of the fatigue range, illustrating that transient
peak size need not follow the monotone equilibrium dependence on
$\omega$. These diagnostics complement, rather than replace, the
equilibrium comparative statics.

\begin{figure}[!ht]
\centering
\includegraphics[
    width=\textwidth,
    keepaspectratio
]{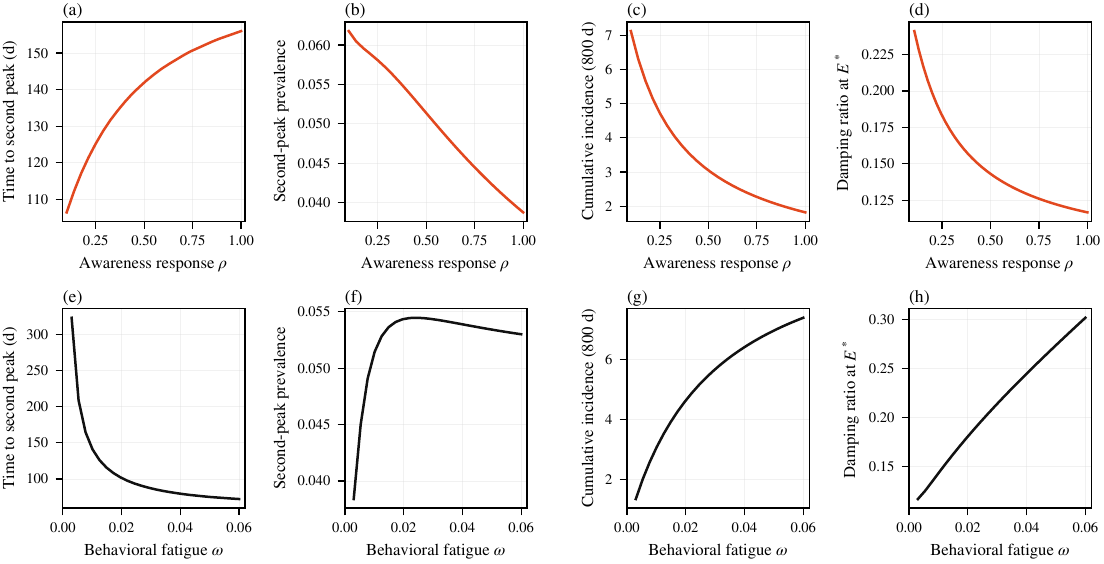}
\caption{\textbf{Quantitative effects of awareness response and fatigue
on transient epidemic dynamics.}
(a--d) Time to the second prevalence peak, second-peak prevalence,
cumulative incidence over $800$ days, and damping ratio at $E^{*}$ as
functions of awareness response $\rho$.
(e--h) Corresponding quantities as functions of behavioral fatigue
$\omega$. The damping ratio in (d) and (h) is $|\operatorname{Re}\lambda|/|\operatorname{Im}\lambda|$
for the leading eigenvalue of the Jacobian restricted to the invariant
hyperplane $\Gam$, which excludes the transverse mode $-\mu$ of
\Cref{sec:hopf}. In (h) that leading eigenvalue is a complex pair over the
whole range $\omega\in[0.003,0.06]$, rising from $0.117$ to $0.302$ with no
change of character. In (d) the leading mode on $\Gam$ is real for
$\rho<0.107$ at $\omega=0.01$, where the damping ratio is not defined, so
the curve is drawn only for $\rho\ge0.107$; panels (a)--(c) run from
$\rho=0.1$. All trajectories use the common initial condition specified
in \Cref{sec:results}; remaining parameters take the baseline values in
\Cref{tab:params}.}
\label{fig:wave_metrics}
\end{figure}

\newpage

\bibliographystyle{apalike}

\end{document}